\documentclass[11pt]{article}

\usepackage[utf8]{inputenc}
\usepackage[T1]{fontenc}
\usepackage{lmodern}
\usepackage[left=1in,right=1in,top=1in,bottom=1.25in]{geometry}

\usepackage{microtype}

\usepackage{amsmath}
\usepackage{amssymb}
\usepackage{amsthm}
\usepackage{thmtools}
\usepackage{thm-restate}

\usepackage{titlesec}
\usepackage{fancybox}
\usepackage{xcolor}
\usepackage{xspace}
\usepackage{enumerate}

\usepackage[ocgcolorlinks]{hyperref} 
\usepackage{cleveref}
\usepackage{authblk}
\usepackage{algorithm}
\usepackage{algorithmic}

\titlespacing{\paragraph}{%
	0pt}{
	0.5\baselineskip}{
	1em}

\makeatletter

\def\moverlay{\mathpalette\mov@rlay}
\def\mov@rlay#1#2{\leavevmode\vtop{%
		\baselineskip\z@skip \lineskiplimit-\maxdimen
		\ialign{\hfil$\m@th#1##$\hfil\cr#2\crcr}}}
\newcommand{\charfusion}[3][\mathord]{
	#1{\ifx#1\mathop\vphantom{#2}\fi
		\mathpalette\mov@rlay{#2\cr#3}
	}
	\ifx#1\mathop\expandafter\displaylimits\fi}
\makeatother

\colorlet{DarkRed}{red!75!black}
\colorlet{DarkGreen}{green!75!black}
\colorlet{DarkBlue}{blue!75!black}

\hypersetup{
	linkcolor = DarkRed,
	citecolor = DarkGreen,
	urlcolor = DarkBlue,
	bookmarksnumbered = true,
	linktocpage = true
}

\declaretheorem[numberwithin=section]{theorem}

\newenvironment{fminipage}%
{\begin{Sbox}\begin{minipage}}%
		{\end{minipage}\end{Sbox}\fbox{\TheSbox}}

\newtheorem{thm}{Theorem}  
\newtheorem{lem}[thm]{Lemma}

\newtheorem{clm}[thm]{Claim}

\renewcommand{\tilde}{\widetilde}

\DeclareMathOperator{\poly}{poly}
\DeclareMathOperator{\OPT}{OPT}
\newcommand{\SelectColumns}{\textbf{\text{ApproximatelySelectColumns}}}

\newcommand{\R}{\mathbb{R}}

\newcommand{\A}{\mathcal{A}}
\newcommand{\BS}{B^{\star}}
\newcommand{\Bs}[1]{B_{#1}^{\star}}
\newcommand{\Ex}{\mathbb{E}}
\newcommand{\opt}{\mathrm{OPT}}
\newcommand{\optone}{\mathrm{OPT}^{(1)}}
\newcommand{\optk}{\mathrm{OPT}^{(k)}}
\newcommand{\eps}{\epsilon}
\newcommand{\nnz}[1]{\left\Vert #1 \right\Vert_0}
\newcommand{\nnzs}[1]{\Vert #1 \Vert_0}
\newcommand{\Def}{\overset{\text{def}}{=}}

\newcommand{\probOne}{Boolean $\ell_0$-Rank-$1$\xspace}

\newcommand{\realsRankK}{Reals $\ell_0$-Rank-$k$\xspace}
\newcommand{\realsRankOne}{Reals $\ell_0$-Rank-$1$\xspace}

\newcommand{\tu}{{\widetilde u}}
\newcommand{\tv}{{\widetilde v}}
\newcommand{\tZ}{{\widetilde Z}}

\begin{document}

\title{Approximation Algorithms for \\ $\ell_0$-Low Rank Approximation}

\author[1]{Karl Bringmann}
\author[1]{Pavel Kolev\thanks{This work has been funded by the Cluster of Excellence ``Multimodal Computing and Interaction'' within the Excellence Initiative of the German Federal Government.}}
\author[2]{David P. Woodruff}

\affil[1]{Max Planck Institute for Informatics, Saarland Informatics Campus, Saarbr{\"u}cken, Germany \protect\\ \texttt{\{kbringma,pkolev\}@mpi-inf.mpg.de}}
\affil[2]{Department of Computer Science, Carnegie Mellon University \protect\\ \texttt{dwoodruf@cs.cmu.edu}}

\date{}	
\maketitle

\begin{abstract}
  We study the $\ell_0$-Low Rank Approximation Problem, where the goal is, 
  given an $m \times n$ matrix $A$, to output a rank-$k$ matrix $A'$ for which
  $\|A'-A\|_0$ is minimized.
  Here, for a matrix $B$, $\|B\|_0$ denotes the number of its non-zero entries. 
  This NP-hard variant of low rank approximation is natural for problems 
  with no underlying metric, and its goal is to minimize the number of disagreeing
  data positions.
  
  We provide approximation algorithms which significantly improve the running time 
  and approximation factor of previous work. 
  For $k > 1$, we show how to find, in poly$(mn)$ time for every $k$, 
  a rank $O(k \log(n/k))$ matrix $A'$ for which $\|A'-A\|_0 \leq O(k^2 \log(n/k)) \OPT$. 
  To the best of our knowledge, this is the first algorithm with provable guarantees 
  for the $\ell_0$-Low Rank Approximation Problem for $k > 1$, 
  even for bicriteria algorithms. 
  
  For the well-studied case when $k = 1$, we give a $(2+\epsilon)$-approximation 
  in {\it sublinear time}, which is impossible for other variants of low rank 
  approximation such as for the  Frobenius norm. 
  We strengthen this for the well-studied case of binary matrices to obtain 
  a $(1+O(\psi))$-approximation in sublinear time, 
  where $\psi = \OPT/\nnz{A}$.
  For small $\psi$, our approximation factor is $1+o(1)$.
\end{abstract}

\newpage
\tableofcontents
\newpage

\section{Introduction}
Low rank approximation of an $m \times  n$ matrix $A$ is an extremely well-studied problem, where the goal is to replace the matrix $A$ with a rank-$k$ matrix $A'$ which well-approximates $A$, in the sense that $\|A-A'\|$ is small under some measure $\|\cdot\|$. Since any rank-$k$ matrix $A'$ can be written as $U \cdot V$, where $U$ is $m \times k$ and $V$ is $k \times n$, this allows for a significant parameter reduction. Namely, instead of storing $A$, which has $mn$ entries, one can store $U$ and $V$, which have only $(m+n)k$ entries in total. Moreover, when computing $Ax$, one can first compute $Vx$ and then $U(Vx)$, which takes $(m+n)k$ instead of $mn$ time. We refer the reader to several surveys \cite{kv09, m11, w14} for references to the many results on low rank approximation. 

We focus on approximation algorithms for the low-rank approximation problem, i.e.  we seek to output a rank-$k$ matrix $A'$ for which $\|A-A'\| \leq \alpha \|A-A_k\|$, where $A_k = \textrm{argmin}_{\textrm{rank}(B)=k} \|A-B\|$ is the best rank-$k$ approximation to $A$, and the approximation ratio $\alpha$ is as small as possible.
One of the most widely studied error measures is the Frobenius norm $\|A\|_F = (\sum_{i=1}^m \sum_{j=1}^n A_{i,j}^2 )^{1/2}$, for which the optimal rank-k approximation can be obtained via the singular value decomposition (SVD). Using randomization and approximation, one can compute an $\alpha = 1+\epsilon$-approximation, for any $\epsilon > 0$, in time much faster than the $\min(m n^2, m n^2)$ time required for computing the SVD, namely, in $O(\nnz{A} + n \cdot \poly(k/\epsilon))$ time \cite{cw13,mm13,nn13}, where $\nnz{A}$ denotes the number of non-zero entries of $A$. For the Frobenius norm $\nnz{A}$ time is also a lower bound, as any algorithm that does not read nearly all entries of $A$ might not read a very large entry, and therefore cannot achieve a relative error approximation.

The rank-$k$ matrix $A_k$ obtained by computing the SVD is also optimal with respect to any rotationally invariant norm, such as the operator and Schatten-$p$ norms. Thus, such norms can also be solved exactly in polynomial time. Recently, however, there has been considerable interest \cite{cw15, agkm16, rsw16} in obtaining low rank approximations for NP-hard error measures such as the {\it entrywise} $\ell_p$-norm $\| A\|_p = \big(\sum_{i,j} |A_{i,j}|^p \big)^{1/p}$, where $p \geq 1$ is a real number. Note that for $p < 1$, this is not a norm, though it is still a well-defined quantity. For $p = \infty$, this corresponds to the max-norm or Chebyshev norm. It is known that one can achieve a $\poly(k \log (mn))$-approximation in $\poly(mn)$ time for the low-rank approximation problem with entrywise $\ell_p$-norm for every $p \geq 1$ \cite{swz16, cgklpw17}.

\subsection{$\ell_0$-Low Rank Approximation}
A natural variant of low rank approximation which the results above do not cover is that of {\it $\ell_0$-low rank approximation}, where the measure $\|A\|_0$ is the number of non-zero entries. In other words, we seek a rank-$k$ matrix $A'$ for which the number of entries $(i,j)$ with $A'_{i,j} \neq A_{i,j}$ is as small as possible. Letting $\OPT = \min_{\textrm{rank}(B)=k} \sum_{i,j} \delta(A_{i,j} \neq A'_{i,j})$, where $\delta(A_{i,j} \neq A'_{i,j}) = 1$ if $A_{i,j} \neq A'_{i,j}$ and $0$ otherwise, we would like to output a rank-$k$ matrix $A'$ for which there are at most $\alpha \OPT$ entries $(i,j)$ with $A'_{i,j} \neq A_{i,j}$. Approximation algorithms for this problem are essential since solving the problem exactly is NP-hard \cite{dajw15,GV15}, even when $k = 1$ and $A$ is a binary matrix.  

The $\ell_0$-low rank approximation problem is quite natural for problems with no underlying metric,
and its goal is to minimize the number of disagreeing data positions with a low rank matrix. 
Indeed, this error measure directly answers the following question: if we are allowed to ignore some data - outliers or anomalies - what is the best low-rank model we can get? One well-studied case is when $A$ is binary, but $A'$ and its factors $U$ and $V$ need not necessarily be binary. This is called unconstrained Binary Matrix Factorization in \cite{j14}, which has applications to association rule mining \cite{k03}, biclustering structure identification \cite{z10,z07}, pattern discovery for gene expression \cite{sjy09}, digits reconstruction \cite{mgnr06}, mining high-dimensional discrete-attribute data \cite{kgr05,k06}, market based clustering \cite{li05}, and document clustering \cite{z07}. There is also a body of work on Boolean Matrix Factorization which restricts the factors to also be binary, which is referred to as constrained Binary Matrix Factorization in \cite{j14}. This is motivated in applications such as classifying text documents and there is a large body of work on this, see, e.g.  \cite{mv14,rpg15}.

The $\ell_0$-low rank approximation problem coincides with a number of problems in different areas. It exactly coincides with the famous matrix rigidity problem over the reals, which asks for the minimal number $\OPT$ of entries of $A$ that need to be changed in order to obtain a matrix of rank at most~$k$. The matrix rigidity problem is well-studied in complexity theory \cite{d76, d80, v77} and parameterized complexity \cite{FLMSZ17}. These works are not directly relevant here as they do not provide approximation algorithms. There are also other variants of $\ell_0$-low rank approximation, corresponding to cases such as when $A$ is binary, $A' = UV$ is required to have binary factors $U$ and $V$, and multiplication is either performed over a binary field \cite{y11,ggyt12,dajw15,prf15}, or corresponds to an OR of ANDs. The latter is known as the Boolean model \cite{bv10,dajw15,mmgdm08,sbm03,sh06,vag07}. These different notions of inner products lead to very different algorithms and results for the $\ell_0$-low rank approximation problem. However, all these models coincide in the special and important case in which $A$ is binary and $k = 1$. This case was studied in \cite{k03,sjy09,j14}, as their algorithm for $k = 1$ forms the basis for their successful heuristic for general $k$, e.g.  the PROXIMUS technique \cite{k03}. 

Another related problem is robust PCA~\cite{c11},
in which there is an 
underlying matrix $A$ that can be written as a low rank matrix $L$ plus a sparse matrix $S$~\cite{Candes2011}. 
Cand\`es et al.~\cite{Candes2011} argue that both components are of arbitrary magnitude, and we do not know the locations 
of the non-zeros in $S$ nor how many there are. Moreover, grossly corrupted observations are common 
in image processing, web data analysis, and bioinformatics where some measurements are \textit{arbitrarily} corrupted due to occlusions, malicious tampering, or sensor failures.
Specific scenarios include video surveillance, face recognition, latent semantic indexing, and ranking of movies, books, etc.~\cite{Candes2011}. 
These problems have the common theme of being an arbitrary magnitude sparse perturbation to a low rank matrix with no natural underlying metric, and so the $\ell_0$-distance measure (which is just the Hamming distance, or number of disagreements) is appropriate.
In order to solve robust PCA in practice, Cand\`es et al.~\cite{Candes2011} relaxed the $\ell_0$-distance measure to the $\ell_1$-norm. Understanding theoretical guarantees for solving the original $\ell_0$-problem is of fundamental importance, and we study this problem in this paper.

Finally, interpreting $0^0$ as $0$, the $\ell_0$-low rank approximation problem coincides with the aforementioned notion of entrywise $\ell_p$-approximation when $p = 0$. It is not hard to see that previous work \cite{cgklpw17} for general $p \geq 1$ fails to give any approximation factor for $p = 0$. Indeed, critical to their analysis is the scale-invariance property of a norm, which does not hold for $p = 0$ since $\ell_0$ is not a norm. 

\subsection{Our Results}
We provide approximation algorithms for the $\ell_0$-low rank approximation problem which significantly improve the running time or approximation factor of previous work. In some cases our algorithms even run in {\it sublinear time}, i.e., faster than reading all non-zero entries of the matrix.
This is provably impossible for other measures such as the Frobenius norm and more generally, any $\ell_p$-norm for $p > 0$. For $k > 1$, our approximation algorithms are, to the best of our knowledge, the first with provable guarantees for this problem. 

First, for $k = 1$, we significantly improve the polynomial running time of previous $(2+\epsilon)$-approximations for this problem. The best previous algorithm due to Jiang et al.~\cite{j14} was based on the observation
that there exists a column $u$ of $A$ spanning a $2$-approximation. Therefore, solving the problem $\min_v \|A-uv\|_0$  for each column $u$ of $A$ yields a 2-approximation, where for a matrix $B$ the measure $\|B\|_0$ counts the number of non-zero entries.
The problem $\min_v \|A-uv\|_0$ decomposes into
$\sum_i \min_{_i} \|A_{:,i}-v_i u\|_0$, where $A_{:,i}$ is the $i$-th column
of $A$, and $v_i$ the $i$-th entry of vector $v$. The optimal $v_i$ is the mode of the ratios
$A_{i,j}/u_j$, where $j$ ranges over indices in $\{1, 2, \ldots, m\}$ with $u_j \neq 0$.
As a result, one can find a rank-1 matrix $u v^T$ providing a $2$-approximation in $O(\nnz{A} n)$ time, which was the
best known running time. Somewhat surprisingly, we show that one can achieve {\it sublinear time}
for solving this problem. Namely, we obtain a $(2+\epsilon)$-approximation 
in $(m+n) \poly(\eps^{-1} \psi^{-1} \log(mn))$ time, for any $\epsilon > 0$, 
where $\psi = \OPT / \nnz{A}$.
This significantly improves upon the
earlier $O( \nnz{A} n )$ time for not too small $\eps$ and $\psi$. 
Our result should be contrasted to Frobenius norm low rank approximation, for which $\Omega(\nnz{A})$ time is required
even for $k = 1$, as otherwise one might miss a very large entry in $A$. Since $\ell_0$-low rank approximation
is insensitive to the magnitude of entries of $A$, we bypass this general impossibility result.

Next, still considering the case of $k = 1$, we show that if the matrix $A$ is binary, a well-studied case
coinciding with the abovementioned $GF(2)$ and Boolean models,
we obtain an approximation algorithm parameterized in terms of the ratio $\psi = \OPT / \nnz{A}$, 
showing it is possible in time $(m+n) \psi^{-1} \poly(\log(mn))$ 
to obtain a $(1+O(\psi))$-approximation. Note
that our algorithm is again sublinear, unlike all algorithms in previous work. Moreover,
when $A$ is itself very well approximated by a low rank matrix, then $\psi$ may actually be sub-constant,
and we obtain a significantly better $(1+o(1))$-approximation than the previous best known $2$-approximations. Thus,
we simultaneously improve the running time and approximation factor. We also show that the running time of our
algorithm is optimal up to $\poly(\log(mn))$ factors by proving that any $(1+O(\psi))$-approximation succeeding with constant
probability must read $\Omega((m+n) \psi^{-1})$ entries of $A$ in the worst case. 

Finally, for arbitrary $k > 1$, we first give an impractical algorithm, running in $n^{O(k)}$ time and achieving an $\alpha = \poly(k)$-approximation. To the best of our knowledge this is the first approximation algorithm for the $\ell_0$-low rank approximation problem with any non-trivial approximation factor. To make our algorithm practical, we reduce the running time to $\poly(mn)$, with an exponent independent of $k$, if we allow for a bicriteria solution. 
In particular, we allow the algorithm to output a matrix $A'$ of somewhat larger rank 
$O(k\log(n/k))$, for which $\|A-A'\|_0 \leq O(k^2 \log(n/k)) \cdot \min_{\textrm{rank}(B)=k}\|A-B\|_0$. 
Although we do not obtain rank exactly $k$, many of the motivations for finding a low rank approximation, such as reducing the number of parameters and fast matrix-vector product, still hold if the output rank is $O(k \log(n/k))$. We are not aware of any alternative algorithms which achieve $\poly(mn)$ time and any provable approximation factor, even for bicriteria solutions.

\section{Preliminaries}

For an matrix $A \in \mathbb{A}^{m \times n}$ with entries $A_{i,j}$, we write $A_{i,:}$ for its $i$-th row and $A_{:,j}$ for its $j$-th column.

\paragraph{Input Formats} 
We always assume that we have random access to the entries of the given matrix~$A$, i.e.  we can read any entry $A_{i,j}$ in constant time. 
For our sublinear time algorithms we need more efficient access to the matrix, specifically the following two variants:

(1) We say that we are given $A$ \emph{with column adjacency arrays} if we are given arrays $B_1,\ldots,B_n$ and lengths $\ell_1,\ldots,\ell_n$ such that for any $1 \le k \le \ell_j$ the pair $B_j[k] = (i,A_{i,j})$ stores the row $i$ containing the $k$-th nonzero entry in column $j$ as well as that entry $A_{i,j}$. This is a standard representation of matrices used in many applications. Note that given only these adjacency arrays $B_1,\ldots,B_n$, in order to access any entry $A_{i,j}$ we can perform a binary search over $B_j$, and hence random access to any matrix entry is in time $O(\log n)$. 
Moreover, we assume to have random access to matrix entries in constant time, and note that this is optimistic by at most a factor $O(\log n)$.

(2) We say that we are given matrix $A$ \emph{with row and column sums} if we can access the numbers $\sum_j A_{i,j}$ for $i \in [m]$ and $\sum_i A_{i,j}$ for $j \in [n]$ in constant time (and, as always, access any entry $A_{i,j}$ in constant time). Notice that storing the row and column sums takes $O(m+n)$ space, and thus while this might not be standard information it is very cheap to store.

We show that the first access type even allows to sample from the set of nonzero entries uniformly in constant time.

\begin{lem} \label{lem:samplingnonzeroentries}
  Given a matrix $A \in \mathbb{R}^{m \times n}$ with column adjacency arrays, after $O(n)$ time preprocessing we can sample a uniformly random nonzero entry $(i,j)$ from $A$ in time $O(1)$.
\end{lem}
\begin{proof}
  Note that we are in particular given the number of nonzero entries $\ell_j = \nnzs{A_{:,j}}$ for each column. We want to first sample a column $X \in [n]$ such that $\Pr[X = j] = \ell_j / \sum_{k \in [n]} \ell_k$, then sample $Y \in [\ell_j]$ uniformly, read $B_X[Y] = (i, A_{i,X})$, and return $A_{i,X}$. Observe that this process indeed samples each nonzero entry of $A$ with the same probability, since the probability of sampling a particular nonzero entry $(i,j)$ is $(\ell_j / \sum_{k \in [n]} \ell_k) \cdot (1/ \ell_j) = 1/\sum_{k \in [n]} \ell_k$. Sampling $Y \in [\ell_j]$ uniformly can be done in constant time by assumption. For sampling $X$, we use the classic Alias Method by Walker~\cite{W74}, which is given the probabilities $\Pr[X=1],\ldots,\Pr[X=n]$ as input and computes, in $O(n)$ time, a data structure that allows to sample from $X$ in time $O(1)$. This finishes the construction.
\end{proof}

\section{Algorithms for Reals $\ell_{0}$-Rank-$k$}

Given a matrix $A\in\mathbb{R}^{m\times n}$, the \realsRankK
problem asks to find a matrix $A^{\prime}\in\mathbb{R}^{m\times n}$ 
with rank $k$ such that the difference between $A$ and $A^{\prime}$ 
measured in $\ell_{0}$-distance is minimized. We denote the optimum value by

\begin{equation}\label{eq:l0-Rank-k}
\optk\overset{\text{def}}{=}\min_{\mathrm{rank}(A^{\prime})=k}\left\Vert A - A^{\prime}\right\Vert _{0}=\min_{U\in\mathbb{R}^{m\times k},\:V\in\mathbb{R}^{k\times n}}\left\Vert A - UV\right\Vert _{0}.
\end{equation}

In this section, we establish several new results on the \realsRankK problem.
In Subsection~\ref{subsec:3.1SR}, we prove a structural lemma that shows the existence 
of $k$ columns which provide a $(k+1)$-approximation to $\optk$, 
and we also give an $\Omega(k)$-approximation lower bound 
for any algorithm that selects $k$ columns from the input matrix $A$.
In Subsection~\ref{subsec:3.2BasicAlg}, we give an approximation algorithm
that runs in $\mathrm{poly}(n^k,m)$ time and achieves an $O(k^2)$-approximation.
To the best of our knowledge, this is the first algorithm with provable non-trivial 
approximation guarantees.
In Subsection~\ref{subsec:3.3BiCritAlg}, we design a practical algorithm that 
runs in $\mathrm{poly}(n,m)$ time with an exponent independent of $k$, 
if we allow for a bicriteria solution.

\subsection{Structural Results}\label{subsec:3.1SR}

We give a new structural result for $\ell_0$-distance showing that 
any matrix $A$ contains $k$ columns which provide a $(k+1)$-approximation 
for the \realsRankK problem.

\begin{lem}\label{lem_SR}
	Let $A\in\mathbb{R}^{m\times n}$
	be a matrix and $k\in\left[n\right]$. 
	There is a subset $J^{(k)}\subset[n]$
	of size $k$ and a matrix $Z\in\mathbb{R}^{k\times n}$ such that
	$\lVert A - A_{:,J^{(k)}}Z\rVert_{0}\leq(k+1)\optk$.
\end{lem}

\begin{proof}
	Let $Q^{(0)}$ be the set of columns $j$ with $U V_{:,j} = 0$, and let $R^{(0)} \overset{\text{def}}{=} [n] \setminus Q^{(0)}$. Let $S^{(0)}\overset{\text{def}}{=}[n]$,
	$T^{(0)}\overset{\text{def}}{=}\emptyset$. We split the value $\optk$ into $\mathrm{OPT}(S^{(0)},R^{(0)})\overset{\text{def}}{=}\lVert A_{S^{(0)},R^{(0)}} - UV_{S^{(0)},R^{(0)}} \rVert_{0}$
	and 
	\[
	\mathrm{OPT}(S^{(0)},Q^{(0)})\overset{\text{def}}{=}\lVert A_{S^{(0)},Q^{(0)}} - UV_{S^{(0)},Q^{(0)}} \rVert_{0}=\lVert A_{S^{(0)},Q^{(0)}}\rVert_{0}.
	\]
	
	Suppose $\mathrm{OPT}(S^{(0)},R^{(0)})\geq|S^{(0)}||R^{(0)}|/(k+1)$.
	Then, for any subset $J^{(k)}$ it follows that 
	\[
	\min_{Z}\lVert A - A_{S^{(0)},J^{(k)}}Z \rVert_{0}\leq|S^{(0)}||R^{(0)}|+\lVert A_{S^{(0)},Q^{(0)}}\rVert_{0}\leq(k+1)\optk.
	\]
	Otherwise, there is a column $i^{(1)}$ such that 
	\[
	\big\Vert A_{S^{(0)},i^{(1)}} -  (UV)_{S^{(0)},i^{(1)}} \big\Vert _{0}\leq\mathrm{OPT}(S^{(0)},R^{(0)})/|R^{(0)}| \le \optk / |R^{(0)}|.
	\]
	
	Let $T^{(1)}$ be the set of indices on which $(UV)_{S^{(0)},i^{(1)}}$
	and $A_{S^{(0)},i^{(1)}}$ disagree, and similarly $S^{(1)}\overset{\text{def}}{=}S^{(0)}\backslash T^{(1)}$
	on which they agree. Then we have $|T^{(1)}|\leq \optk / |R^{(0)}|$. Hence, in the submatrix $T^{(1)} \times R^{(0)}$ the total error is at most $|T^{(1)}| \cdot |R^{(0)}| \le \optk$.
	Let $R^{(1)},D^{(1)}$ be a partitioning
	of $R^{(0)}$ such that $A_{S^{(1)},j}$ is linearly dependent on
	$A_{S^{(1)},i^{(1)}}$ iff $j\in D^{(1)}$. Then by selecting
	column $A_{:,i^{(1)}}$ the incurred cost on matrix $S^{(1)}\times D^{(1)}$ is zero. 
	For the remaining submatrix $S^{(\ell)}\times R^{(\ell)}$, we perform a recursive call of the algorithm.
	
	We make at most $k$ recursive calls, on instances $S^{(\ell)}\times R^{(\ell)}$
	for $\ell\in\{0,\ldots,k-1\}$. In the $\ell^{th}$ iteration, either
	$\mathrm{OPT}(S^{(\ell)},R^{(\ell)})\geq|S^{(\ell)}||R^{(\ell)}|/(k+1-\ell)$
	and we are done, or there is a column $i^{(\ell+1)}$ which partitions
	$S^{(\ell)}$ into $S^{(\ell+1)},T^{(\ell+1)}$ and $R^{(\ell)}$
	into $R^{(\ell+1)},D^{(\ell+1)}$ such that 
	\[
	|S^{(\ell+1)}|\geq m\cdot\prod_{i=0}^{\ell}\left(1-\frac{1}{k+1-i}\right)=
	\frac{k-\ell}{k+1}\cdot m
	\]
	and for every $j\in D^{(\ell)}$ the column $A_{S^{(\ell+1)},j}$
	belongs to the span of $\{A_{S^{(\ell+1)},i^{(t)}}\}_{t=1}^{\ell+1}$.
	
	Suppose we performed $k$ recursive calls. We show now that
	the incurred cost in submatrix $S^{(k)}\times R^{(k)}$ 
	is at most $\mathrm{OPT}(S^{(k)},R^{(k)})\leq\optk$.
	By construction, $|S^{(k)}|\geq m/(k+1)$ and the sub-columns
	$\{A_{S^{(k)},i}\}_{i\in I^{(k)}}$ are linearly independent, 
	where $I^{(k)}=\{i^{(1)},\dots,i^{(k)}\}$ is the set of the selected columns,
	and $A_{S^{(k)},I^{(k)}}=(UV)_{S^{(k)},I^{(k)}}$.
	Since $\mathrm{rank}(A_{S^{(k)},I^{(k)}})=k$, it follows that
	$\mathrm{rank}(U_{S^{(k)},:})=k$, $\mathrm{rank}(V_{:,I^{(k)}})=k$ and
	the matrix $V_{:,I^{(k)}}\in\mathbb{R}^{k\times k}$ is invertible.
	Hence, for matrix $Z=(V_{:,I^{(k)}})^{-1}V_{:,R^{k}}$ we have
	$$\mathrm{OPT}(S^{(k)},R^{(k)})=\lVert A_{S^{k},R^{k}}-A_{S^{k},I^{k}}Z\rVert_{0}.$$
	
	The statement follows by noting that the recursive calls accumulate 
	a total cost of at most $k\cdot\optk$ in the submatrices
	$T^{(\ell+1)}\times R^{(\ell)}$ for $\ell \in \{0,1,\ldots,k-1\}$, as well as cost at most $\optk$ in submatrix $S^{(k)}\times R^{(k)}$.
\end{proof}

We also show that any algorithm that selects $k$ columns of 
a matrix $A$ incurs at least an $\Omega(k)$-approximation 
for the \realsRankK problem.

\begin{lem}\label{lem_lowerBound_L0rankK}
	Let $k \le n/2$.
	Suppose $A=(G_{k\times n};\ I_{n\times n})\in\mathbb{R}^{(n+k)\times n}$
	is a matrix composed of a Gaussian random matrix $G\in\mathbb{R}^{k\times n}$
	with $G_{i,j}\sim N(0,1)$ and identity matrix $I_{n\times n}$. Then
	for any subset $J^{(k)}\subset\left[n\right]$ of size $k$, we have
	$\min_{Z\in\mathbb{R}^{k\times n}}\lVert A - A_{:,J^{(k)}}Z\rVert_{0}=\Omega(k) \cdot \optk.$
\end{lem}
\begin{proof}
	Notice that the optimum cost is at most $n$, achieved by selecting 
	$U=(I_{k\times k};\ 0_{n\times k})$ and $V=G_{k\times n}$. 
	It is well known that Gaussian matrices are invertible with probability 1, 
	see e.g. \cite[Thm 3.3]{SST06}. Hence, $G_{:,J^{(k)}}$ is a nonsingular matrix 
	for every subset $J^{(k)}\subset[n]$ of size $k$.

	We will show next that for any subset $J^{(k)}$ of $k$ columns 
	the incurred cost is at least $(n-k)k \ge nk/2$.
	Without loss of generality, the chosen columns $J^{(k)} = [k]$ are the first $k$ columns of $A$. Let $R = [2k]$ be the first $2k$ rows and $C = [n] \setminus J$ be the last $n-k$ columns. We bound 
	\begin{align*}
	\min_Z \nnzs{A - A_{:,[k]} Z} 
	& \ge \min_Z \nnzs{A_{R,C} - A_{R,[k]} Z} \\
	& = \sum_{j \in C} \min_{z^{(j)}} \nnzs{A_{R,j} - A_{R,[k]} z^{(j)}},
	\end{align*}
	i.e. we ignore all rows and columns except $R$ and $C$. Consider any column $j \in C$. Since $A_{R,j} = (G_{:,j}, 0_{k})$ and $A_{R,[k]} = (G_{:,[k]}, I_{k \times k})$, for any vector $z \in \mathbb{R}^k$ we have
	\begin{align*}
	\nnzs{A_{R,j} - A_{R,J^{(k)}} z} 
	& = \nnzs{G_{:,j} - G_{:,[k]} z} + \nnzs{I_{k \times k} z} \\
	& = \nnzs{G_{:,j} - G_{:,[k]} z} + \nnz{z}.
	\end{align*}
	Let $\ell \Def \nnz{z}$. By symmetry, without loss of generality we can assume that the first $\ell$ entries of $z$ are non-zero and the remaining entries are 0. Let $x \in \mathbb{R}^\ell$ be the vector containing the first $\ell$ entries of $z$. Then we have
	$$\nnzs{A_{R,j} - A_{R,J^{(k)}} z} = \nnzs{G_{:,j} - G_{:,[\ell]} x} + \ell.$$

	We consider w.l.o.g. the first $k$ columns of $A$,
	and we construct the optimum matrix $Z$ that minimizes
	$\lVert A_{:,1:k}Z - A \rVert_{0}$. Observe that it is 
	optimal to set the first $k$ columns of $Z$ to $I_{k\times k}$,
	and since $A_{2k+1:n,1:k}=0$ we can focus only on the submatrix
	$A_{1:2k,k+1:n}=(G_{1:k,k+1:n}; 0_{k\times n-k})$.
	
	Consider a column $A_{1:2k,j}$ for $j\in[k+1,n]$. Our goal is to find
	a vector $v\in\mathbb{R}^k$ minimizing the objective function
	$\Psi=\min_{v}\{\lVert v\rVert_{0}+\lVert G^{(k)}v-g\rVert_{0}\}$,
	where $G^{(k)}\Def\{G_{1:k,1:k}\}$ and $g\Def G_{1:k,j}$.
	It holds with probability $1$ that $G^{(k)}$ and $g$
	do not have an entry equal to zero. Moreover, since $G^{(k)}$ is invertible
	every row in $[G^{(k)}]^{-1}$ is non-zero, and thus with probability
	$1$ a vector $v=[G^{(k)}]^{-1}g$ has entry equal to zero.

	Let $v=(x;\ 0)$ be an arbitrary vector with $\lVert x\rVert_{1}=\ell$.
	Let $G^{(\ell)}$ be a submatrix of $G^{(k)}$ induced by the first
	$\ell$ columns. For every subset $S\subset[m]$ of $\ell$ rows the
	corresponding submatrix $G_{S,:}^{(\ell)}$ has a full rank. Suppose
	there is a subset $S$ such that for $G_{S,:}^{(\ell)}$ and $g_{S}$
	there is a vector $x\in\mathbb{R}^{k}$ satisfying $G_{S,:}^{(\ell)}x=g_{S}$.
	Since $G_{S,:}^{(\ell)}$ is invertible, the existence of $x$ implies
	its uniqueness. On the other hand, for any row $i\in[m]\backslash S$
	the probability of the event $G_{i,:}^{(\ell)}x=g_{i}$ is equals
	to $0$. Since $G^{(k)}v=G^{(\ell)}x$ and there are finitely many possible 
	subsets $S$ as above, i.e. ${m \choose \ell}\leq m^{\ell}$, by union bound 
	it follows that $\lVert G^{(k)}v-g\rVert_{0}\geq k-\ell$.
	Therefore, it holds that $\phi\geq k$.
	
	The statement follows by noting that the total cost incurred by
	$A_{:,1:k}$ and any $Z$ is lower bounded by
	$(n-k)k+(n-k)=\left(1-k/n\right)(k+1)n$.
\end{proof}

\subsection{Basic Algorithm}\label{subsec:3.2BasicAlg}

We give an impractical algorithm that runs in $\mathrm{poly}(n^k,m)$ 
time and achieves an $O(k^2)$-approximation. To the best of our knowledge this is 
the first approximation algorithm for the \realsRankK problem
with non-trivial approximation guarantees.
\begin{thm}\label{thm_L0rankK_k^2approx} 
	Given $A\in\mathbb{R}^{m\times n}$ and $k\in[n]$ we can compute in
	$O(n^{k+1}m^{2}k^{\omega+1})$ time a set of $k$ indices 
	$J^{(k)}\subset[n]$ and a matrix 
	$Z\in\mathbb{R}^{k\times n}$ such that
	$\lVert A - A_{:,J^{(k)}}Z\rVert_{0}\leq O(k^{2}) \cdot \optk$.
\end{thm}

We use as a subroutine the algorithm of Berman and Karpinski~\cite{bk02} 
(attributed also to Kannan in that paper) which given a matrix $U$ and a vector $b$
approximates $\min_x \|Ux-b\|_0$ in polynomial time.
Specifically, we invoke in our algorithm the following variant of this result 
established by Alon, Panigrahy, and Yekhanin~\cite{APY09}.

\begin{thm}\label{thm_kAppxAlg} \cite{APY09} 
	There is an algorithm that given a matrix $A\in\mathbb{R}^{m\times k}$
	and a vector $b\in\mathbb{R}^{m}$, outputs in time $O(m^{2}k^{\omega+1})$
	a vector $z\in\mathbb{R}^{k}$ such that w.h.p.\ $\left\Vert Az-b\right\Vert _{0}\leq k\cdot\min_{x}\left\Vert Ax-b\right\Vert _{0}$.
\end{thm}

\begin{proof}[Proof of Theorem~\ref{thm_L0rankK_k^2approx}]
	The existence of a subset $J^{*}$ of $k$ columns of $A$ and matrix 
	$Z^* \in\mathbb{R}^{k\times n}$
	with $\lVert A - A_{:,J^{*}}Z^* \rVert_{0}\leq(k+1)\optk$
	follows by Lemma \ref{lem_SR}. 
	We enumerate all ${n \choose k}$ subsets $J^{(k)}$ of $k$ columns.
	For each $J^{(k)}$, we split $\min_{Z}\lVert A_{:,J^{(k)}}Z-A\rVert _{0}=
	\sum_{i=1}^{n}\min_{z^{(i)}}
	\lVert A_{:,J^{(k)}}z^{(i)}-A_{:,i}\rVert _{0}$, and we run the algorithm from 
	Theorem~\ref{thm_kAppxAlg} for each column $A_{:,i}$, obtaining approximate 
	solutions $\tilde z^{(1)},\ldots,\tilde z^{(n)}$ that form a matrix $\tilde Z$. 
	Then, we return the best solution $(A_{:,J^{(k)}},\tilde Z)$. 
	To verify that this yields a $k(k+1)$-approximation, 
	note that for $J^{(k)} = J^*$ we have
	\begin{eqnarray*}
		\lVert A_{:,J^{*}}\tilde{Z}-A\rVert_{0}
		&=&\sum_{i=1}^{n}\lVert A_{:,J^{*}}\tilde{z}^{(i)}-A_{:,i}\rVert_{0}\le k\sum_{i=1}^{n}\min_{z^{(i)}}\lVert A_{:,J^{*}}z^{(i)}-A_{:,i}\rVert_{0}\\
		&=&k\cdot\min_{Z}\lVert A_{:,J^{*}}Z-A\rVert_{0}
		\le k(k+1)\cdot\optk.
	\end{eqnarray*}
	
	The time bound $O(n^{k+1} m^2 k^{\omega+1})$ is immediate from Theorem~\ref{thm_kAppxAlg}. 
	This proves the statement.
\end{proof}

\subsection{Bicriteria Algorithm}\label{subsec:3.3BiCritAlg}

Our main contribution in this section is to design a practical algorithm
that runs in $\mathrm{poly}(n,m)$ time with an exponent independent of $k$, 
if we allow for a bicriteria solution.
\begin{thm}\label{thm_manuscript} 
	Given $A\in\mathbb{R}^{m\times n}$ with $m\geq n$ and $k\in\{1,\dots,n\}$, 
	there is an algorithm that in expected $\mathrm{poly}(m,n)$ time outputs a subset of indices
	$J\subset[n]$ with $|J|=O(k\log (n/k))$ and a matrix $Z\in\mathbb{R}^{|J|\times n}$
	such that $\left\Vert A - A_{:,J}Z\right\Vert _{0}\leq O(k^{2} \log (n/k))\cdot \optk$.
\end{thm}
The structure of the proof follows a recent approximation 
algorithm~\cite[Algorithm 3]{cgklpw17} for the $\ell_p$-low rank approximation problem, for any $p \geq 1$. We note that the analysis of \cite[Theorem 7]{cgklpw17} is missing an $O(\log^{1/p} n)$ approximation factor, and na\"ively provides an $O(k\log^{1/p}n)$-approximation rather than the stated 
$O(k)$-approximation. Further, it might be possible to obtain an efficient algorithm yielding an $O(k^2\log k)$-approximation for Theorem~\ref{thm_manuscript} using unpublished techniques 
in \cite{swz18}; we leave the study of obtaining the optimal approximation factor to future work.

There are two critical differences with the proof of \cite[Theorem 7]{cgklpw17}.
We cannot use the earlier \cite[Theorem 3]{cgklpw17} which shows that any 
matrix $A$ contains $k$ columns which provide an $O(k)$-approximation 
for the $\ell_p$-low rank approximation problem, since that proof requires
$p \geq 1$ and critically uses scale-invariance, which does not hold for $p = 0$. 
Our combinatorial argument in Lemma~\ref{lem_SR} seems fundamentally
different than the maximum volume submatrix argument in \cite{cgklpw17} 
for $p \geq 1$. 

Second, unlike for $\ell_p$-regression for $p \geq 1$, 
the $\ell_0$-regression problem $\min_x \|Ux-b\|_0$ given a matrix $U$ and vector $b$ is
not efficiently solvable since it corresponds to a nearest codeword problem,
which is NP-hard~\cite{Alekhnovich11a}.
Thus, we resort to an approximation algorithm for $\ell_0$-regression, 
based on ideas for solving the nearest codeword problem in \cite{APY09,bk02}.

Note that $\optk \leq \nnz{A}$. Since there are only $mn+1$ possibilities
of $\optk$, we can assume we know $\optk$ and we can run the 
Algorithm~\ref{alg:randomcolumns} below for each such possibility, 
obtaining a rank-$O(k \log n)$ solution,
and then outputting the solution found with the smallest cost.

\begin{algorithm}[t]
	\SelectColumns($A, k$)
	
	$\quad$ ensure $A$ has at least $k \log (n/k))$ columns
	\smallskip\smallskip
	
	$\,\,\,$1. \textbf{If} the number of columns of matrix $A$ is less than or equal to $2k$
	
	$\,\,\,$2. $\quad$ \textbf{Return} all the columns of $A$
	
	$\,\,\,$3. \textbf{Else} 
	
	$\,\,\,$4. $\quad$ \textbf{Repeat}
	
	$\,\,\,$5. $\quad$ $\quad$ Let $R$ be a set of $2k$ uniformly random columns of $A$
	
	$\,\,\,$6. $\quad$ \textbf{Until} at least $1/10$ fraction of columns of $A$ 
	are nearly approximately covered
	\smallskip
	
	$\,\,\,$7. $\quad$ Let $A_{\overline{R}}$ be the columns of $A$ not
	nearly approximately covered by $R$
	
	$\,\,\,$8. $\quad$ \textbf{Return} $R \cup \SelectColumns(A_{\overline{R}}, k)$
	
	\caption{Bicriteria Algorithm: Selecting $O(k \log (n/k))$ Columns}
	\label{alg:randomcolumns}
\end{algorithm}

This can be further optimized by forming instead $O(\log(mn))$ guesses of $\optk$. 
One of these guesses is within a factor of $2$ from the true value of $\optk$, 
and we note that the following argument only needs to know $\optk$ up to a factor of $2$. 

We start by defining the notion of approximate coverage, which is different than the corresponding notion in \cite{cgklpw17} for
$p \geq 1$, due to the fact that $\ell_0$-regression cannot be efficiently solved. Consequently, approximate coverage for $p = 0$ cannot be efficiently tested.
Let $Q\subseteq [n]$ and $M=A_{:,Q}$ be an $m\times |Q|$ submatrix of $A$. 
We say that a column $M_{:,i}$ is
$(S,Q)$-\emph{approximately covered} by a submatrix $M_{:,S}$ of $M$, if $|S|=2k$ and 
\begin{equation}
\min_x \|M_{:,S}\cdot x-M_{:,i}\|_0 \leq \frac{100(k+1) \optk}{|Q|}.
\end{equation}

\begin{lem}(Similar to \cite[Lemma 6]{cgklpw17}, but using Lemma \ref{lem_SR})\label{lem:cover}
	Let $Q\subseteq [n]$ and $M=A_{:,Q}$ be a submatrix of $A$.
	Suppose we select a subset $R$ of $2k$ uniformly random columns of $M$. 
	Then with probability at least $1/3$, at least a $1/10$ fraction of the columns of $M$
	are $(R,Q)$-approximately covered.
\end{lem}
\begin{proof}
	To show this, as in \cite{cgklpw17}, 
	consider a uniformly random column index $i$ not in the set $R$.
	Let $T \Def R \cup \{i\}$, $\eta \Def \min_{\mathrm{rank}(B)=k} \nnz{M_{:,T}-B}$,
	and $\BS\Def\arg\min_{\mathrm{rank}(B)=k} \nnz{M-B}$. 
	Since $T$ is a uniformly random subset of $2k+1$ columns of $M$, we have
	\begin{eqnarray*}
		\Ex_{T}\eta
		& \leq &
		\Ex_{T} \nnz{M_{:,T} - \Bs{:,T}}
		= \sum_{T\in{|Q| \choose 2k+1}}
		\sum_{i\in 	T}\nnz{M_{:,i} - \Bs{:,i}} Pr\left[T\right] \\
		& = & \sum_{i\in Q}\frac{{|Q|-1 \choose |T-1|}}{{|Q| \choose |T|}}
		\nnz{M_{:,i} - \Bs{:,i}}
		= \frac{(2k+1)\optk_M}{|Q|} \leq \frac{(2k+1)\optk}{|Q|}.
	\end{eqnarray*}
	Then, by a Markov bound, we have $\Pr[\eta \leq \tfrac{10(2k+1)\optk}{|Q|}]\geq9/10$. 
	Let $\mathcal{E}_1$ denotes this event.
	
	Fix a configuration $T=R\cup\{i\}$ and let $L(T)\subset T$ be the subset guaranteed by
	Lemma~\ref{lem_SR} such that $|L(T)|=k$ and 
	\[
	\min_{X}\nnz{M_{:,L(T)}X-M_{:,T}}\leq\left(k+1\right)\min_{\mathrm{rank}(B)=k}\nnz{M_{:,T}-B}.
	\]
	Notice that
	\[
	\Ex_{i}\left[\min_{x}\nnz{M_{:,L(T)}x-M_{:,i}}\ |\ 	T\right]=\frac{1}{2k+1}\min_{X}\nnz{M_{:,L(T)}X-M_{:,T}},
	\]
	and thus by the law of total probability we have
	\[
	\Ex_{T}\left[\min_x \|M_{:,L(T)}x-M_{:,i}\|_0\right] \leq \frac{(k+1) \eta}{2k+1}.
	\]
	Let $\mathcal{E}_2$ denote the event that 
	$\min_x \|M_{:,L}x-M_{:,i}\|_0 \leq \tfrac{10(k+1)\eta}{2k+1}$.
	By a Markov bound, $\Pr[\mathcal{E}_2] \geq 9/10$.
	
	Further, as in \cite{cgklpw17}, let $\mathcal{E}_3$ be the event that 
	$i \notin L$.
	Observe that there are ${k+1\choose k}$ ways to choose a subset $R'\subset T$ 
	such that $|R'|=2k$ and $L\subset R'$. 
	Since there are ${2k+1 \choose 2k}$ ways to choose $R'$, it follows that
	$\Pr[L\subset R\ | \ T]=\tfrac{k+1}{2k+1}>1/2$.
	Hence, by the law of total probability, we have $\Pr[\mathcal{E}_3]>1/2$.
	
	As in \cite{cgklpw17}, 
	$\Pr[\mathcal{E}_1 \wedge \mathcal{E}_2 \wedge \mathcal{E}_3] > 2/5$,
	and conditioned on $\mathcal{E}_1 \wedge \mathcal{E}_2 \wedge \mathcal{E}_3$,
	\begin{equation}
	\min_x \|M_{:,R}x-M_{:,i}\|_0 \leq \min_x \|M_{:,L}x-M_{:,i}\|_0 \leq \frac{10(k+1) \eta}{2k+1} \leq \frac{100(k+1)\optk}{|Q|},
	\end{equation}
	where the first inequality uses that $L$ is a subset of $R$ given $\mathcal{E}_3$, and so the regression cost cannot decrease, while the
	second inequality uses the occurrence of $\mathcal{E}_2$ and the final inequality uses the occurrence of $\mathcal{E}_1$.
	
	As in \cite{cgklpw17}, if $Z_i$ is an indicator random variable indicating 
	whether $i$ is approximately covered by $R$, and $\tZ = \sum_{i\in Q} Z_i$, then
	$\Ex_R [\tZ] \geq \tfrac{2|Q|}{5}$ and $\Ex_R [|Q|-\tZ] \leq \tfrac{3|Q|}{5}$. 
	By a Markov bound, it follows that $\Pr[|Q|-\tZ \geq \tfrac{9|Q|}{10}] \leq \tfrac{2}{3}$.
	Thus, probability at least $1/3$, at least a $1/10$ fraction of the columns of $M$ are $(R,Q)$-approximately covered.
\end{proof}

Given Lemma~\ref{lem:cover}, we are ready to prove Theorem~\ref{thm_manuscript}.
As noted above, a key difference with the corresponding \cite[Algorithm 3]{cgklpw17} 
for $\ell_p$ and $p \geq 1$, is that we cannot efficiently test
if the $i$-th column is approximately covered by the set $R$.
We will instead again make use of Theorem~\ref{thm_kAppxAlg}.

\begin{proof}[Proof of Theorem~\ref{thm_manuscript}]
	The computation of matrix $Z$ force us to relax the notion of $(R,Q)$-approximately covered 
	to the notion of $(R,Q)$-nearly-approximately covered as follows: 
	we say that a column $M_{:,i}$
	is $(R,Q)$-\emph{nearly-approximately covered} if, 
	the algorithm in Theorem~\ref{thm_kAppxAlg} returns a vector $z$ such that 
	\begin{equation}\label{eq:NAC}
	\|M_{:,R}z-M_{:,i}\|_0 \leq \frac{100(k+1)^2\optk}{|Q|}.
	\end{equation}
	By the guarantee of Theorem~\ref{thm_kAppxAlg}, if $M_{:,i}$ is $(R,Q)$-approximately
	covered then it is also with probability at least $1-1/\poly(mn)$ 
	$(R,Q)$-nearly-approximately covered.
	
	Suppose Algorithm~\ref{alg:randomcolumns} makes $t$ iterations and
	let $A_{:,\cup_{i=1}^{t}R_{i}}$ and $Z$ be the resulting solution.
	We bound now its cost. 
	Let $B_{0}=[n]$, and consider the $i$-th iteration of Algorithm~\ref{alg:randomcolumns}. 
	We denote by $R_i$ a set of $2k$ uniformly random columns 
	of $B_{i-1}$, by $G_{i}$ a set of columns that is $(R_{i},B_{i-1})$-nearly-approximately covered,
	and by $B_{i}=B_{i-1}\backslash \{G_{i}\cup R_{i}\}$ a set of 
	the remaining columns. By construction, $|G_{i}|\geq|B_{i-1}|/10$ and 
	$$|B_{i}|\leq\frac{9}{10}|B_{i-1}|-2k<\frac{9}{10}|B_{i-1}|.$$
	Since Algorithm~\ref{alg:randomcolumns} terminates when $B_{t+1}\leq 2k$, we have
	$$2k<|B_{t}|<\left(1-\frac{1}{10}\right)^{t}|B_{0}|=\left(1-\frac{1}{10}\right)^{t}n,$$
	and thus the number of iterations $t\leq 10\log(n/2k)$.
	By construction, $|G_{i}|=(1-\alpha_{i})|B_{i-1}|$ for some $\alpha_{i}\leq9/10$,
	and hence
	\begin{equation}\label{eq:sum}
	\sum_{i=1}^{t}\frac{|G_{i}|}{|B_{i-1}|}\leq t\leq 10\log\frac{n}{2k}.
	\end{equation}
	Therefore, the solution cost is bounded by
	\begin{eqnarray*}
		&&\nnz{A_{:,\cup_{i=1}^{t}R_{i}}Z-A}=\sum_{i=1}^{t}\sum_{j\in G_{i}}\nnz{A_{:,R_{i}}z^{(j)}-A_{:,j}}\\&\overset{\text{Lem.}\ref{lem:cover}}{\leq}&\sum_{i=1}^{t}\sum_{j\in G_{i}}k\cdot\min_{x^{(j)}}\nnz{A_{:,R_{i}}x^{(j)}-A_{:,j}}\overset{\eqref{eq:NAC}}{\leq}\sum_{i=1}^{t}\sum_{j\in G_{i}}\frac{100\left(k+1\right)^{2}\optk}{|B_{i-1}|}\\&=&100\left(k+1\right)^{2}\optk\cdot\sum_{i=1}^{t}\frac{|G_{i}|}{|B_{i-1}|}\overset{\eqref{eq:sum}}{\leq} O\left(k^{2}\cdot\log\frac{n}{2k}\right)\cdot\optk.
	\end{eqnarray*}
	
	By Lemma \ref{lem:cover}, the expected number of iterations of selecting a set $R_{i}$ such that
	$|G_{i}|\geq1/10|B_{i-1}|$ is $O(1)$.
	Since the number of recursive calls $t$ is bounded by $O(\log(n/k))$, it follows by
	a Markov bound that Algorithm~\ref{alg:randomcolumns} chooses $O(k \log(n/k))$ columns in total.
	Since the approximation algorithm of Theorem~\ref{thm_kAppxAlg} runs in polynomial time, 
	our entire algorithm has expected polynomial time.
\end{proof}

\section{Algorithm for Reals $\ell_{0}$-Rank-$1$}

Given a matrix $A\in\mathbb{R}^{m\times n}$, the \realsRankOne problem 
asks to find a matrix $A^{\prime}\in\mathbb{R}^{m\times n}$  with rank $1$,
i.e. $A^{\prime}=uv^T$ for some vectors $u\in\R^{m}$ and $v\in\R^{n}$,
such that the difference between $A$ and $A^{\prime}$ measured 
in $\ell_{0}$-distance is minimized. We denote the optimum value by
\begin{equation}\label{eq:l0-Rank-1}
	\optone\Def 
	\min_{\mathrm{rank}(A^{\prime})=1}\left\Vert A - A^{\prime}\right\Vert_{0}=
	\min_{u\in\mathbb{R}^{m},\:v\in\mathbb{R}^{n}} \nnzs{A-uv^{T}}.
\end{equation}
In the trivial case when $\optone=0$, there is an optimal algorithm that runs 
in time $O(\nnz{A})$ and finds the exact rank-1 decomposition $uv^T$ of a matrix $A$.
In this work, we focus on the case when $\optone\geq1$. We show that Algorithm~\ref{alg_Rand2}
yields a $(2+\eps)$-approximation factor and runs in nearly linear time in $\nnz{A}$, 
for any constant $\eps > 0$.
Furthermore, a variant of our algorithm even runs in sublinear time, if $\nnz{A}$ is large and 
\begin{equation}\label{eq:psi}
	\psi \Def \optone / \nnz{A}
\end{equation}
is not too small. In particular, we obtain time $o(\nnz{A})$ when 
$\optone \ge (\epsilon^{-1}\log (mn))^{4}$ and $\nnz{A} \ge n (\epsilon^{-1}\log (mn))^{4}$.

\begin{thm}\label{thm_ImprRank1}
There is an algorithm that, given $A\in\mathbb{R}^{m\times n}$ 
with column adjacency arrays and $\optone \geq1$, 
and given $\epsilon\in(0,0.1]$,	runs w.h.p.\ in time
\[
	O\left(\Big(
	\frac{n\log m}{\eps^{2}}+
	\min\Big\{\nnz{A},\ n+\psi^{-1}\frac{\log n}{\eps^{2}}\Big\}\Big)
	\frac{\log^{2}n}{\eps^{2}}\right)
\]
and outputs a column $A_{:,j}$ and a vector $z\in\R^n$ such that w.h.p.\
$\lVert A-A_{:,j}\cdot z^{T}\rVert{}_{0}\leq(2+\epsilon)\optone$.
The algorithm also computes an estimate $Y$ satisfying w.h.p.\ 
$(1-\eps)\optone \le Y \le (2+2\eps) \optone$. 	
\end{thm}

In fact, our analysis of Theorem~\ref{thm_ImprRank1} directly applies to the \probOne problem,
for the definition see Section~\ref{sec:BinL0R1withSmallOpt}, 
and yields as a special case the following result
(which we use to prove Theorem~\ref{thm_paramApproxAlgCombined} in Section~\ref{sec:BinL0R1withSmallOpt}).

\begin{thm}\label{thm:BinaryRankOneThereeApprox}
	Let $\opt = \min_{u\in\{ 0,1\}^m,\,v\in\{ 0,1\}^n} \nnzs{A-u\cdot v^{T}}$.
	Given a binary matrix $A\in\{0,1\}^{m\times n}$ with column adjacency arrays and 
	$\opt \geq 1$, and given $\epsilon\in(0,0.1]$,	we can compute w.h.p.\ in time
	\[
		O\left(\Big(
		\frac{n\log m}{\eps^{2}}+\min\Big\{\nnz{A},\ n+\psi^{-1}\frac{\log n}{\eps^{2}}\Big\}
		\Big)\frac{\log^{2}n}{\eps^{2}}\right)
	\]
	a column $A_{:,j}$ and a binary vector $z\in\{0,1\}^{n}$ such that w.h.p.\
	$\lVert A-A_{:,j}\cdot z^{T}\rVert{}_{0}\leq(2+\epsilon)\opt$.
	Further, we can compute an estimate $Y$ such that w.h.p.\ 
	$(1-\eps)\opt \le Y \le (2+2\eps) \opt$. 
\end{thm}

The rest of this section is devoted to proving Theorem~\ref{thm_ImprRank1}.
We present the pseudocode of this result in Algorithm~\ref{alg_Rand2}.

\begin{algorithm}[H]
	\textbf{Input:} $A\in\mathbb{R}^{m\times n}$ and $\epsilon\in\left(0,0.1\right)$.\\
	
	1. Partition the columns of $A$ into weight-classes 
	$\mathcal{S} = \{S^{(0)},\ldots,S^{(1+\log n)}\}$ such that
	
	$\qquad i)$ $S^{(0)}$ contains all columns $j$ with $\nnzs{A_{:,j}} = 0$, and 
	
	$\qquad ii)$ $S^{(i)}$ contains all columns $j$ with $2^{i-1} \leq\lVert A_{:,j}\rVert_{0}<2^{i}$.
	
	2. \textbf{For each} weight-class $S^{(i)}$ do:
	
	$\quad\quad2.1$ Sample a set $C^{(i)}$ of $\Theta(\epsilon^{-2}\log n)$ elements uniformly at random from $S^{(i)}$.
	
	$\quad\quad2.2$ Find a $(1+\tfrac \epsilon {15})$-approximate solution $z^{(j)}\in\mathbb{R}^{n}$
	for each column $A_{:,j}\in C^{(i)}$, i.e.
	\begin{equation}\label{eq:target}
		\big\Vert A-A_{:,j}\cdot [z^{(j)}	]^{T}\big\Vert _{0}\leq
		\left(1+\frac \epsilon {15}\right)\min_{v}\big\Vert A-A_{:,j}\cdot v^{T}\big\Vert _{0}.
	\end{equation}

	3. Compute a $(1+\tfrac \epsilon {15})$-approximation $Y_{j}$ of 
	$\lVert A-A_{:,j}\cdot [z^{(j)}]^{T}\rVert_{0}$ 
	for every $j\in \bigcup_{i\in [|\mathcal{S}|]} C^{(i)}$.
	
	4. \textbf{Return} the pair $(A_{:,j}, z^{(j)})$ corresponding to
	the minimal value $Y_{j}$.
	
	\caption{Reals $\ell_{0}$-Rank-$1$: Approximation Scheme}
	\label{alg_Rand2}
\end{algorithm}

The only steps for which the implementation details are not immediate are Steps 2.2 and 3. We will discuss them in Sections~\ref{sec:steptwowo} and \ref{sec:stepthree}, respectively.

Note that the algorithm from Theorem~\ref{thm_ImprRank1} selects a column $A_{:,j}$ and then finds a good vector $z$ such that the product $A_{:,j} z^T$ approximates $A$. We show that the approximation guarantee $2+\eps$ is essentially tight for algorithms following this pattern.

\begin{lem}\label{lem_LB_l0-Rank-1} 
	There exist a matrix $A\in\mathbb{R}^{n\times n}$
	such that $\min_{z\in\R^{n}}\lVert A-A_{:,j}\cdot z^{T}\rVert_{0}\geq2(1-1/n)\mathrm{OPT}^{(1)}$,
	for every column $A_{:,j}$.
\end{lem}
\begin{proof}[Proof of Lemma~\ref{lem_LB_l0-Rank-1}]
	Let $A=I+J\in\mathbb{R}^{n\times n}$, where $I$ is an identity matrix
	and $J=\mathbf{1}\mathbf{1}^{T}$ is an all-ones matrix. Note that
	$\mathrm{OPT}^{(1)} \le n$ is achieved by approximating $A$ with
	the rank-1 matrix~$J$. On the other hand, when we choose $u=A_{:,i}$
	for any $i\in[n]$, the incurred cost on any column $A_{:,j}$, $j \ne i$,
	is $\min_{x}\lVert A_{:,j} - x A_{:,i}\rVert_{0}=2$, since there are two entries where $A_{:,i}$ and $A_{:,j}$ disagree. Hence,
	the total cost is at least $2n-2\geq(2-2/n)\mathrm{OPT}^{(1)}$.
\end{proof}

\subsection{Correctness}

We first prove the following structural result, capturing Steps 1-2.2 of Algorithm~\ref{alg_Rand2}. 

\begin{lem}\label{lem:existCj}
	Let $C^{(0)},\ldots,C^{(\log n+1)}$ be the sets constructed in Step 2.1 of Algorithm~\ref{alg_Rand2}, and let $C = C^{(0)} \cup \ldots \cup C^{(\log n+1)}$. 
	Then, w.h.p.\ $C$ contains an index $j\in[n]$ such that 
	\[
	\min_{z\in\R^{n}} \nnzs{A - A_{:,j}\cdot z^T} \le (2+\eps/2)\cdot \optone.
	\]
\end{lem}

\begin{proof}
	Let $u,v$ be an optimum solution of \eqref{eq:l0-Rank-1}. 
	For the weight class $S^{(0)}$ containing all columns without nonzero entries, setting $z_c = 0$ for any $c \in S^{(0)}$ gives zero cost on these columns, no matter what column $A_{:,j}$ we picked.
	Hence, without loss of generality in the following we assume that $S^{(0)} = \emptyset$.
	
	For any $i \ge 1$, we partition the weight class $S^{(i)}$ into
	$N^{(i)},Z^{(i)}$ such that $v_{i}=0$ for every $i\in Z^{(i)}$
	and $v_{i}\neq0$ for $i\in N^{(i)}$. 
	We denote by $\mathcal{S}^{+}$ the set of weight-classes $S^{(i)}$
	with $|N^{(i)}|\geq \tfrac 13 |S^{(i)}|$. 
	Let $\mathcal{R}=\bigcup_{i\in\mathcal{S}^{+}}S^{(i)}$
	and $\mathcal{W}=[n]\backslash\mathcal{R}$. 
	We partition $\mathcal{R}=\mathcal{N}\cup\mathcal{Z}$
	such that $v_{i}=0$ for every $i\in\mathcal{Z}$ and $v_{i}\neq0$
	for $i\in\mathcal{N}$. 
	Further, using the three sets $\mathcal{N},\mathcal{Z}$ and $\mathcal{W}$
	we decompose $\optone$ into
	\begin{eqnarray*}
		\optone 
		& = & \mathrm{OPT}_{\mathcal{N}}+\mathrm{OPT}_{\mathcal{Z}}+\mathrm{OPT}_{\mathcal{W}}\\
		& = & 
		\lVert A_{:,\mathcal{N}}-u\cdot v_{\mathcal{N}}^{T}\rVert_{0}+
		\lVert A_{:,\mathcal{Z}}\rVert_{0}+\lVert A_{:,\mathcal{W}}-u
		\cdot v_{\mathcal{W}}^{T}\rVert_{0}.
	\end{eqnarray*}
	The proof proceeds by case distinction:		
	\paragraph{The set $\mathcal{Z}$:}
	For any column $A_{:,j}$ of $A$, we have 
	\begin{equation}\label{eq:optZ}
		\min_{z_{\mathcal{Z}}}\lVert A_{:,\mathcal{Z}}-A_{:,j}\cdot z_{\mathcal{Z}}^{T}\rVert_{0}\leq \lVert A_{:,\mathcal{Z}}-A_{:,j}\cdot 0\rVert_{0} =
		\lVert A_{:,\mathcal{Z}}\rVert_{0}=\mathrm{OPT}_{\mathcal{Z}}.
	\end{equation}
	
	\paragraph{The set $\mathcal{W}$:}
	Note that $\mathcal{W}$ consists of all weight classes $S^{(i)}$ with $|Z^{(i)}| > \tfrac 23 |S^{(i)}|$. For any such weight class $S^{(i)}$, the optimum cost satisfies 
	\[
	\lVert A_{:,S^{(i)}}-uv_{S^{(i)}}^{T}\rVert_{0}\geq\lVert A_{:,Z^{(i)}}\rVert_{0}\geq \frac 23 |S^{(i)}|2^{i-1} = \frac 13 |S^{(i)}| 2^i.
	\]
	Further, for any column $A_{:,j}$ of $A$, we have
	\begin{align*}
		\min_{z} \nnzs{A_{:,S^{(i)}} - A_{:,j} z^T} 
		& \le \nnzs{A_{:,S^{(i)}}} \le \nnzs{A_{:,Z^{(i)}}} + \frac 13 |S^{(i)}| 2^i \\
		& \le 2 \nnzs{A_{:,Z^{(i)}}} \le 2 \lVert A_{:,S^{(i)}}-uv_{S^{(i)}}^{T}\rVert_{0},
	\end{align*}
	and thus the total cost in $\mathcal{W}$ is bounded by
	\begin{equation}\label{eq:optW}
		\min_{z_{\mathcal{W}}}\lVert A_{:,\mathcal{W}}-A_{:,j}\cdot z_{\mathcal{W}}^{T}\rVert_{0}
		=\sum_{i\in\mathcal{W}}\min_{z_{S^{(i)}}}\lVert A_{:,S^{(i)}}-A_{:,j}\cdot z_{S^{(i)}}^{T}\rVert_{0}
		\leq2\lVert A_{:,\mathcal{W}}-uv_{\mathcal{W}}^{T}\rVert_{0}
		=2\,\mathrm{OPT}_{\mathcal{W}}.
	\end{equation}
	
	\paragraph{The set $\mathcal{N}$:} 
	By an averaging argument there is a subset $G\subseteq \mathcal{N}$ of
	size $|G|\geq \tfrac \epsilon 3 |\mathcal{N}|$ such that for every $j\in G$ 
	we have
	\[
		\lVert A_{:,j}-v_{j}\cdot u\rVert_{0}\leq \frac 1{1-\epsilon/3}\cdot \frac{\mathrm{OPT}_{\mathcal{N}}}{|\mathcal{N}|} 
		\le \left(1+\frac \epsilon 2\right)\cdot \frac{\mathrm{OPT}_{\mathcal{N}}}{|\mathcal{N}|}.
	\]	
	Let $j\in G$ be arbitrary. Furthermore, let $P^{(j)}$ be the set of all rows $i$ with $A_{i,j} = v_j \cdot u_i$, and let $Q^{(j)} = [m] \setminus P^{(j)}$. 
	By construction, we have $|Q^{(j)}|\leq(1+\tfrac \epsilon 2)\mathrm{OPT}_{\mathcal{N}}/|\mathcal{N}|$.
	Moreover, since $j \in \mathcal{N}$ we have $v_j \ne 0$, and thus we may choose $z_{\mathcal{N}} = \tfrac 1{v_j} v_{\mathcal{N}}$. This yields
	\begin{align}
		\min_{z_{\mathcal{N}}}\lVert A_{:,\mathcal{N}}-A_{:,j}\cdot z_{\mathcal{N}}^{T}\rVert_{0}
		&\leq\lVert A_{:,\mathcal{N}}-A_{:,j}\cdot\frac{1}{v_{j}}v_{\mathcal{N}}^{T}\rVert_{0}\nonumber\\
		&=\lVert A_{P^{(j)},\mathcal{N}}-u_{P^{(j)}}\cdot v_{\mathcal{N}}^{T}\rVert_{0}+\lVert A_{Q^{(j)},\mathcal{N}}-A_{:,j}\cdot\frac{1}{v_{j}}v_{\mathcal{N}}^{T}\rVert_{0}\nonumber\\
		&\le\mathrm{OPT}_{\mathcal{N}}+\left(1+\frac{\epsilon}{3}\right)
		\mathrm{OPT}_{\mathcal{N}}=\left(2+\frac{\epsilon}{3}\right)\mathrm{OPT}_{\mathcal{N}}.
		\label{eq:optN}
	\end{align}
	
	Hence, by combining \eqref{eq:optZ}, \eqref{eq:optW} and \eqref{eq:optN}, 
	it follows for any index $j\in G$ that
	\begin{align*}
		&\quad\min_{z}\lVert A-A_{:,j}\cdot z^{T}\rVert_{0}\nonumber\\
		&=
		\min_{z_{\mathcal{W}}}\lVert A_{:,\mathcal{W}}-A_{:,j}
		\cdot z_{\mathcal{W}}^{T}\rVert_{0}+
		\min_{z_{\mathcal{Z}}}\lVert A_{:,\mathcal{Z}}-A_{:,j}
		\cdot z_{\mathcal{Z}}^{T}\rVert_{0} + 
		\min_{z_{\mathcal{N}}}\lVert A_{:,\mathcal{N}}-A_{:,j}
		\cdot z_{\mathcal{N}}^{T}\rVert_{0}  \\
		&\le 2\,\mathrm{OPT}_{\mathcal{W}}+
		\mathrm{OPT}_{\mathcal{Z}}+
		\left(2+\frac{\epsilon}{2}\right)\mathrm{OPT}_{\mathcal{N}}
		\le\left(2+\frac{\epsilon}{2}\right)\optone.
	\end{align*}
	This yields the desired approximation guarantee, 
	provided that we sampled a column $A_{:,j}$ from $G$.
	We show next that whenever $\mathcal{N} \ne \emptyset$,
	our algorithm samples with high probability at least one column from $G$.
	
	Before that let us consider the case when $\mathcal{N} = \emptyset$.
	Then, since the bounds \eqref{eq:optZ} and \eqref{eq:optW} hold 
	for any column $A_{:,j}$ of $A$, the set $C$ contains only good columns.
	Thus, we may assume that $\mathcal{N} \ne \emptyset$.
	
	We now analyze the probability of sampling a column $A_{:,j}$ from $G$. 
	By construction, the set $\mathcal{N}$ is the union of all $N^{(i)}$ such that 
	$|N^{(i)}| \ge \tfrac 13 |S^{(i)}|$. As shown above, we have 
	$|G|\geq \tfrac \epsilon 3 |\mathcal{N}|$, and thus there is an index $i$ satisfying
	$|G \cap S^{(i)}| \ge \tfrac \epsilon 3 |N^{(i)}| \ge \tfrac \epsilon 9 |S^{(i)}|$. 
	Hence, when sampling a uniformly random element from $S^{(i)}$ we hit $G$ with probability 
	at least $\tfrac \epsilon 9$. 
	Since we sample $\Theta(\eps^{-2} \log n)$ elements from $S^{(i)}$, 
	we hit $G$ with high probability. This finishes the proof.
\end{proof}

\paragraph{Correctness Proof of Algorithm~\ref{alg_Rand2}: }
It remains to show that the pair $(A_{:,j},z^{j})$ with minimum estimate $Y_{j}$ 
yields a $(2+\epsilon)$-approximation to $\optone$.
By Step 3, for every column $j$ we have
\begin{equation}\label{eq:Ybounds}
	\left(1+\frac{\eps}{15}\right)^{-1}\cdot\lVert A-A_{:,j}[z^{(j)}]^{T}\rVert_{0}\le Y_{j}\le\left(1+\frac{\eps}{15}\right)\cdot\lVert A-A_{:,j}[z^{(j)}]^{T}\rVert_{0}.
\end{equation}
Since $Y_j \le Y_{j'}$ for any other column $j'$, (\ref{eq:Ybounds}) and the approximation 
guarantee of Steps 2.2 yield
\[
	\left(1+\frac{\epsilon}{15}\right)^{-1}\lVert A-A_{:,j}[z^{(j)}]^{T}\rVert_{0}\le\left(1+\frac{\epsilon}{15}\right)\lVert A-A_{:,j'}[z^{(j')}]^{T}\rVert_{0}\le\left(1+\frac{\epsilon}{15}\right)^{2}\min_{z}\big\Vert A-A_{:,j'}z^{T}\big\Vert_{0}.
\]
By Lemma~\ref{lem:existCj}, w.h.p.\ there exists a column $j' \in C$ with 
$\min_z \nnzs{A - A_{:,j'} z^T} \le (2+ \tfrac \eps 2) \optone$. 
We obtain a total approximation ratio of 
$(1+\tfrac \epsilon {15})^3 (2+ \tfrac \eps 2) \le 2+\eps$ for any error $0 < \eps \le 0.1$, 
i.e. we have $\lVert A-A_{:,j}[z^{(j)}]^{T}\rVert_{0} \le (2+\eps) \optone$.
Therefore, it holds that
\[
	(1-\eps)\optone\le\left(1+\frac{\epsilon}{15}\right)^{-1}\optone\le Y_{j}\le\left(1+\frac{\epsilon}{15}\right)(2+\eps)\optone\le(2+2\eps)\optone.
\]
This finishes the correctness proof.

\subsection{Implementing Step 2.2} \label{sec:steptwowo}

Step 2.2 of Algorithm~\ref{alg_Rand2} uses the following sublinear procedure.

\begin{algorithm}[H]
	\textbf{Input:} $A\in\mathbb{R}^{m\times n}$, $u\in\mathbb{R}^{m}$
	and $\epsilon\in\left(0,1\right)$.
	
	$\quad$ let $t\Def \Theta(\epsilon^{-2}\log m)$, $N\Def \mathrm{supp}(u)$,
	and $p\Def t/|N|$.
	\medskip
	
	1. Select each index $i\in N$ with probability $p$ and let $S$
	be the resulting set.
	
	2. Compute a vector $z\in\mathbb{R}^{n}$ such that $z_{j}=\arg\min_{r\in\mathbb{R}}\left\Vert A_{S,j}-r\cdot u_{S}\right\Vert _{0}$
	for all $j\in[n]$.
	
	3. \textbf{Return} the vector $z$.
	
	\caption{Reals $\ell_0$-Rank-1: Objective Value Estimation}
	\label{alg_Rand}
\end{algorithm}

We prove now the correctness of Algorithm~\ref{alg_Rand} and we analyze its runtime.

\begin{lem}\label{lem_RArank1} 
	Given $A\in\mathbb{R}^{m\times n}$ with $m\geq n$, $u\in\mathbb{R}^{m}$ and 
	$\epsilon\in\left(0,1\right)$ we can compute in time
	$O\left(\epsilon^{-2}n\log m\right)$ a vector 
	$z\in\mathbb{R}^{n}$ such that w.h.p. for every index $i\in[n]$ it holds that
	\[
		\lVert A_{:,i}-z_{i}\cdot u\rVert_{0}\leq
		(1+\epsilon)\min_{v_{i}\in\R}\lVert A_{:,i}-v_{i}\cdot u\rVert_{0}.
	\]
\end{lem}

\begin{proof}
	Let $N,Z$ be a partitioning of $[m]$ such that $u_{i}=0$ for
	$i\in Z$ and $u_i \neq 0$ for $i \in N$. Since $\lVert A-u\cdot z^{T}\rVert_{0}=\lVert A_{N,:}-u_{N}\cdot z^{T}\rVert_{0}+\lVert A_{Z,:}\rVert_{0}$,
	it suffices to find a vector $z$ such that for every $j\in[n]$ we have
	\begin{equation}
		\left\Vert A_{N,j}-z_{j}\cdot u_{N}\right\Vert _{0}\leq\left(1+\epsilon\right)\cdot\min_{v_{j}}\left\Vert A_{N,j}-v_{j}\cdot u_{N}\right\Vert _{0}.\label{eq:showMe}
	\end{equation}
	Let $j\in [n]$ be arbitrary. 
	For $r \in \mathbb{R}$ let $G(r) \Def \{ i \in N \,:\, A_{i,j} / u_i = r \}$ be the set of entries of $A_{N,j}$ that we correctly recover by setting $z_j = r$. 
	Note that $\nnzs{A_{N,j} - z_j \cdot u_N} = |N| - |G(z_j)|$ holds for any $z_j \in \mathbb{R}$.
	Hence, the optimal solution sets $z_j = r^{\star} \Def  \arg\max_{r \in \mathbb{R}} |G(r)|$. 
	
	Let $X_{G(r)}$ be a random variable
	indicating the number of elements selected from group $G(r)$ in Step 1 of Algorithm~\ref{alg_Rand}.
	Notice that $\mathbb{E}[X_{G(r)}]=t \cdot |G(r)|/|N|$,
	and by Chernoff bound w.h.p.\ we have 
	\begin{equation}\label{eq:ExpXGr}
		|X_{G(r)}-\mathbb{E}[X_{G(r)}]|\leq(\epsilon/8)\cdot t.
	\end{equation}
	Let $S\subseteq N$ be the set of selected indices. 
	Further, since $|S|=\sum_{\ell}X_{G(r)}$ and $\mathbb{E}[|S|]=t$, 
	by Chernoff bound we have w.h.p.\ $|S|\leq(1+\eps)|t|$.
	Observe that Step 2 of Algorithm~\ref{alg_Rand} selects 
	$z_j = \arg\max_{r\in\mathbb{R}}X_{G(r)}$, since
	$\nnzs{A_{S,j}-r\cdot u_{S}} = |S| - X_{G(r)}$. 
	We now relate $z_j$ to $r^{\star}$. The proof proceeds by case distinction 
	on $\delta^{\star} \Def  |G(r^{\star})| / |N|$.
	\medskip
	
	\textbf{Case 1:} Suppose $\delta^{\star} \leq \eps/4$. 
	Then $\nnzs{A_{N,j} - r \cdot u_N} \ge (1-\eps/4) |N|$ 
	for every $r\in\mathbb{R}$, and thus no matter which $z_j$ is selected 
	we obtain a $(1+\eps)$-approximation, since
	\[
		\nnzs{A_{N,j}-z_{j}\cdot u_{N}}\le|N|\le(1+\eps)\min_{r}\nnzs{A_{N,j}-r\cdot u_{N}}.
	\]
	
	\textbf{Case 2:} Suppose $\delta^{\star} \ge 1/2 + \eps$. 
	Then, by \eqref{eq:ExpXGr} w.h.p. we have
	\[
		X_{G(r^{\star})} \ge \Ex[X_{G(r^{\star})}] - (\eps/8) t = 
		(\delta^{\star} - \eps/8) t > (1+\eps) t/2 \ge |S| / 2,
	\]
	and thus $X_{r^{\star}}$ is maximal among all $X_r$. 
	Hence, we select the optimal $z_j = r^{\star}$.
	\medskip
	
	\textbf{Case 3:} Suppose $\eps/4 < \delta^{\star} < 1/2 + \eps$. 
	Let $z_j = r$ be the value chosen by Algorithm~\ref{alg_Rand}.
	By \eqref{eq:ExpXGr}, the event of making a mistake,
	given by $X_{G(r)}\geq X_{G(r^{\star})}$, happens when
	\begin{equation}\label{eq:GrGrstar}
		\Ex[X_{G(r)}] + (\eps/8) t \ge \Ex[X_{G(r^{\star})}] - (\eps/8) t.
	\end{equation}
	Let $\delta \Def  |G(r)| / |N|$ and note that (\ref{eq:GrGrstar}) implies 
	$\delta \ge \delta^{\star} - \eps/4$. 
	Hence, for the selected $r\neq r^{\star}$ we have
	\begin{align*}
		\nnzs{A_{N,j} - r \cdot u_N} &= 
		(1-\delta) |N| \le (1-\delta^{\star} + \eps/4) |N| \\
		& \le (1+\eps) (1-\delta^{\star}) |N| = (1+\eps) \nnzs{A_{N,j} - r^{\star} \cdot u_N}.
	\end{align*}
	
	Therefore, in each of the preceding three cases, we obtain w.h.p.\ 
	a $(1+\eps)$-approximate solution. The statement follows by the union bound.	
\end{proof}

\subsection{Implementing Step 3} \label{sec:stepthree}

In Step 3 of Algorithm~\ref{alg_Rand2}, our goal is to compute 
a $(1+\tfrac \epsilon {15})$-approximation $Y_{j}$ of $\lVert A-A_{:,j}[z^{(j)}]^{T}\rVert_{0}$,
for every $j\in \bigcup_{i\in [|\mathcal{S}|]} C^{(i)}$.

In this subsection, our main algorithmic result is the following.

\begin{thm}\label{thm:samplingPQ}
	There is an algorithm that, given $A \in \mathbb{R}^{m \times n}$ 
	with column adjacency arrays and $\optone\geq1$, 
	and given $j \in [n]$, $v \in \mathbb{R}^m$ and $\epsilon \in (0,1)$,
	outputs an estimator $Y$ that satisfies w.h.p.\
	\[
	(1-\epsilon) \nnz{A-A_{:,j}\cdot v^{T}} \le Y 
	\le (1+\epsilon)\nnz{ A-A_{:,j}\cdot v^{T}}.
	\] 
	The algorithm runs w.h.p.\ in time 
	$O(\min\{\nnz{A}, n+\epsilon^{-2} \psi^{-1}\log n\})$, where $\psi = \optone / \nnz{A}$.
\end{thm}

We prove Theorem~\ref{thm:samplingPQ}, by designing: 
i) an exact deterministic algorithm, see Lemma~\ref{lem:exactAAizT}; and
ii) a randomized approximation algorithm running in sublinear-time, see Lemma~\ref{lem:samplingAB_0}.

\begin{lem}	\label{lem:exactAAizT}
Suppose $A,B\in\mathbb{R}^{m\times n}$ are represented by column
adjacency arrays. Then, we can compute in $O(\lVert A\rVert_{0}+n)$
time the measure $\lVert A-B\rVert_{0}$.
\end{lem}
\begin{proof}
We partition the entries of $A$ into five sets:
\begin{align*}
&T_{1}=\left\{ (i,j)\,:\,A_{ij}=0\text{ and }B_{ij}\neq0\right\},
&&T_{4}=\left\{ (i,j)\,:\,0\neq A_{ij}=B_{ij}\neq0\right\},\\
&T_{2}=\left\{ (i,j)\,:\,A_{ij}\neq0\text{ and }B_{ij}=0\right\},
&&T_{5}=\left\{ (i,j)\,:\,A_{ij}=B_{ij}=0\right\},\\
&T_{3}=\left\{ (i,j)\,:\,0\neq A_{ij}\neq B_{ij}\neq0\right\}.
\end{align*}

Observe that $\lVert A-B\rVert_{0}=|T_{1}|+|T_{2}|+|T_{3}|$ and $\lVert B\rVert_{0}=|T_{1}|+|T_{3}|+|T_{4}|$.
Since
$\lVert A-B\rVert_{0}=\lVert B\rVert_{0}+|T_{2}|-|T_{4}|$,
it suffices to compute the numbers $\lVert B\rVert_{0}$, $|T_{2}|$
and $|T_{4}|$. We compute $|T_{2}|$ and $|T_{4}|$ in
$O(\lVert A\rVert_{0})$ time, by enumerating all non-zero entries 
of $A$ and performing $O(1)$ checks for each.
For $\lVert B\rVert_{0}$, we sum the column lengths 
of $B$ in $O(n)$ time.
\end{proof}

For our second, sampling-based implementation of Step 3, we make use of an algorithm 
by Dagum et al.~\cite{DKLR00} for estimating the expected value of a random variable.
We note that the runtime of their algorithm is a random variable, the magnitude of 
which is bounded w.h.p.\ within a certain range.

\begin{thm}\cite{DKLR00}\label{thm:AAsampleAlg}
	Let $X$ be a random variable taking values in $[0,1]$ with $\mu \Def \mathbb{E}[X]>0$.
	Let $0 < \epsilon,\delta < 1$
	and $\rho_{X}=\max\{\mathrm{Var}[X],\epsilon\mu\}$.
	There is
	an algorithm with sample access to $X$ that computes an estimator $\tilde \mu$ in time $t$ such
	that for a universal constant $c$ we have
	\smallskip\smallskip
	
	i) $\Pr[(1-\epsilon)\mu\leq \tilde \mu\leq(1+\epsilon)\mu]\geq1-\delta$
	$\quad\text{ and }\quad$
	ii) $\Pr[t\geq c\cdot\epsilon^{-2}\mu^{-2}\rho_{X}\log(1/\delta)]\leq\delta$.
\end{thm}

We state now our key technical insight, on which we build upon our sublinear algorithm.

\begin{lem}\label{lem:samplingAB_0}
	There is an algorithm that, given $A,B\in\mathbb{R}^{m\times n}$ 
	with column adjacency arrays and $\nnzs{A-B}\geq1$, and given $\epsilon>0$,
	computes an estimator $Z$ that satisfies w.h.p.\
	$$(1-\eps) \nnzs{A-B} \le Z \le (1+\eps) \nnzs{A-B}.$$
	The algorithm runs w.h.p. in time 
	$O(n+\epsilon^{-2}\tfrac{\lVert A\rVert_{0}+
		\lVert B\rVert_{0}}{\lVert A-B\rVert_{0}}\log n\})$.
\end{lem}
\begin{proof}
	By Lemma~\ref{lem:samplingnonzeroentries}, after $O(n)$ preprocessing time 
	we can sample a uniformly random non-zero entry from $A$ or $B$ in $O(1)$ time.
	
	We consider the following random process:
	
	1. Sample $C\in\{A,B\}$ such that $\Pr[C=A]=\tfrac{\lVert A\rVert_{0}}{\lVert A\rVert_{0}+\lVert B\rVert_{0}}$
	and $\Pr[C=B]=\tfrac{\lVert B\rVert_{0}}{\lVert A\rVert_{0}+\lVert B\rVert_{0}}$.
	
	2. Sample $(i,j)$ uniformly at random from the non-zero entries of $C$
	
	3. Return:
	\[
	X=\begin{cases}
	0, & \text{if }A_{ij}=B_{ij};\\
	1/2, & \text{if }0\neq A_{ij}\neq B_{ij}\neq0;\\
	1, & \text{if }A_{ij}\neq B_{ij}\text{ and either }A_{ij}\text{ or }B_{ij}\text{ equals }0.
	\end{cases}
	\]
	Observe that
	\begin{eqnarray*}
		\mathbb{E}[X] & = & \sum_{(i,j)\colon A_{ij}\neq B_{ij} = 0}\frac{\lVert A\rVert_{0}}{\lVert A\rVert_{0}+\lVert B\rVert_{0}}\cdot\frac{1}{\lVert A\rVert_{0}}+\sum_{(i,j)\colon 0 = A_{ij}\neq B_{ij}}\frac{\lVert B\rVert_{0}}{\lVert A\rVert_{0}+\lVert B\rVert_{0}}\cdot\frac{1}{\lVert B\rVert_{0}}\\
		&  & +\sum_{(i,j)\colon 0\neq A_{ij}\neq B_{ij}\neq0}\left(\frac{1}{2}\cdot\frac{\lVert A\rVert_{0}}{\lVert A\rVert_{0}+\lVert B\rVert_{0}}\cdot\frac{1}{\lVert A\rVert_{0}}+\frac{1}{2}\cdot\frac{\lVert B\rVert_{0}}{\lVert A\rVert_{0}+\lVert B\rVert_{0}}\cdot\frac{1}{\lVert B\rVert_{0}}\right)\\
		& = & \frac{\lVert A-B\rVert_{0}}{\lVert A\rVert_{0}+\lVert B\rVert_{0}}.
	\end{eqnarray*}
	Straightforward checking shows that $X\in\left[0,1\right]$ implies 
	$\mathrm{Var}[X]\leq\mathbb{E}[X]$, and thus 
	\[
	\rho_X = \max\{\mathrm{Var}[X],\epsilon\cdot \mathbb{E}[X]\} \le \mathbb{E}[X].
	\]
	Setting $\delta = 1/\mathrm{poly}(n)$ in Theorem~\ref{thm:AAsampleAlg}, we can compute w.h.p.\ in time
	$O(\epsilon^{-2} \Ex[X]^{-1} \log n) = O(\epsilon^{-2}\tfrac{\lVert A\rVert_{0}+
		\lVert B\rVert_{0}}{\lVert A-B\rVert_{0}}\log n)$
	an estimator 
	$(1-\eps) \mathbb{E}[X] \le \tilde \mu \le (1+\eps) \mathbb{E}[X]$. 
	Then, w.h.p.\ the estimator $Z \Def (\nnzs{A} + \nnzs{B}) \tilde\mu$ satisfies the statement.
\end{proof}

We are now ready to prove Theorem~\ref{thm:samplingPQ}.

\begin{proof}[Proof of Theorem~\ref{thm:samplingPQ}]
	Let $B \Def A_{:,j} v^T$ and observe that $\nnz{A-B}\geq\optone\geq1$.
	Note that we implicitly have column adjacency arrays for $B$, 
	since for any column $c$ with $v_c = 0$ there are no non-zero entries in $B_{:,c}$, 
	and for any column $c$ with $v_c = 1$ the non-zero entries of $B_{:,c}$ are 
	the same as for $A_{:,j}$. 
	Hence, Lemma~\ref{lem:exactAAizT} and Lemma~\ref{lem:samplingAB_0} are applicable.
	
	We analyze the running time of Lemma~\ref{lem:samplingAB_0}. 
	Note that if $\lVert B\rVert _{0}\leq(1+\psi)\lVert A\rVert _{0}$ then
	$\tfrac{\lVert A\rVert_{0}+\lVert B\rVert_{0}}{\lVert A-B\rVert_{0}}\leq(2+\psi)/\psi$,
	and otherwise, i.e. $(1+\psi)\lVert A\rVert _{0}\leq\lVert B\rVert _{0}$, we have
	$$\nnz{A-B} \ge \nnz{B} - \nnz{A} \ge \tfrac{\psi}{1+\psi} \nnz{B}$$ and thus 
	$\tfrac{\lVert A\rVert_{0}+\lVert B\rVert_{0}}{\lVert A-B\rVert_{0}} \leq 2(1+\psi)/\psi$.
	Hence, $\tfrac{\lVert A\rVert_{0}+\lVert B\rVert_{0}}{\lVert A-B\rVert_{0}} < 4/\psi$, 
	which yields with high probability time $O(n+\epsilon^{-2} \psi^{-1}\log n)$.
	
	We execute in parallel the algorithms from Lemma~\ref{lem:exactAAizT} 
	and Lemma~\ref{lem:samplingAB_0}. Once the faster algorithm outputs a solution,
	we terminate the execution of the slower one. 
	Note that this procedure runs w.h.p\ in time $O(\min\{\nnz{A}, n+\epsilon^{-2} \psi^{-1}\log n\})$,
	and returns w.h.p.\ the desired estimator $Y$.
\end{proof}

To implement Step 3 of Algorithm~\ref{alg_Rand2}, we simply apply Theorem~\ref{thm:samplingPQ} 
with $A$,  $\epsilon$ and $v = z^{(j)}$ to each sampled column $j \in \bigcup_{0 \le i \le \log n +1} C^{(i)}$.

\subsection{Analyzing the Runtime of Algorithm~\ref{alg_Rand2}}

Consider again Algorithm~\ref{alg_Rand2}.
In Steps 1, 2 and 2.1, from each of the $O(\log n)$ weight classes we sample $O(\epsilon^{-2}\log n)$ columns. 
In Step 2.2, for each sampled column we use Lemma~\ref{lem_RArank1}, which takes time 
$O(\eps^{-2} n \log m)$ per column, or $O(\eps^{-4} n \log m \log^2 n)$ in total. 
Finally in Step 3, for each sampled column we use Theorem~\ref{thm:samplingPQ}, which w.h.p.\ takes time $O(\min\{\nnz{A}, n+\epsilon^{-2} \psi^{-1}\log n\})$ per column, or 
$O(\min\{\nnz{A},n+\epsilon^{-2} \psi^{-1}\log n\} \cdot \eps^{-2} \log^2 n)$ in total.
Then, the total runtime is bounded by
$O(\eps^{-4} n \log m \log^2 n + \min\{ \nnz{A} \eps^{-2} \log^2 n,\, \epsilon^{-4} \psi^{-1}\log^3 n\})$.

\section{Algorithms for Boolean $\ell_{0}$-Rank-$1$}\label{sec:BinL0R1withSmallOpt}

Given a matrix $A\in\{0,1\}^{m\times n}$, the \probOne problem 
asks to find a matrix $A^{\prime}\in\{0,1\}^{m\times n}$  with rank $1$,
i.e. $A^{\prime}=uv^T$ for some vectors $u\in\{0,1\}^{m}$ and $v\in\{0,1\}^{n}$,
such that the difference between $A$ and $A^{\prime}$ measured
in $\ell_{0}$-distance is minimized. We denote the optimum value by

\begin{equation}\label{eq:goal}
\opt = \opt_A \;\Def\; \min_{u\in\{ 0,1\}^m,\,v\in\{ 0,1\}^n} \nnz{A-uv^{T}}.
\end{equation}

In practice, approximating a matrix $A$ by a rank-1 matrix $u v^T$ makes most sense if $A$ is close to being rank-1. Hence, the above optimization problem is most relevant in the case $\opt \ll \nnz{A}$. 
For this reason, in this section we focus on the case $\opt / \nnz{A} \le \phi$ for sufficiently 
small $\phi > 0$. We prove the following.

\begin{thm} \label{thm_paramApproxAlg}
Given $A\in\{ 0,1\}^{m\times n}$ with row and column sums, 
and given $\phi \in (0,\tfrac{1}{80}]$ with $\opt / \nnz{A} \le \phi$, 
we can compute in time
$O(\min\{\nnz{A} + m+n, \phi^{-1} (m+n) \log(mn)\})$
vectors $\tu\in\{0,1\}^{m}$ and $\tv\in\{0,1\}^{n}$ 
such that w.h.p.\
$\nnzs{A - \tu\cdot \tv^T} \le (1+5\phi) \opt + 37 \phi^2 \nnz{A}$.
\end{thm}

Using Theorem~\ref{thm:BinaryRankOneThereeApprox}, we can compute a $(2+\eps)$-approximation of $\opt$, and thus a $(2+\eps)$-approximation of the ratio $\opt / \nnz{A}$. 
Hence, combining Theorem~\ref{thm_paramApproxAlg} and Theorem~\ref{thm:BinaryRankOneThereeApprox}, 
yields an algorithm that does not need $\phi$ as an input
and computes a $(1+500\psi)$-approximate solution of the \probOne problem.

\begin{thm} \label{thm_paramApproxAlgCombined}
	Given $A\in\{ 0,1\}^{m\times n}$ with column adjacency arrays and with row and column sums, 
	for $\psi = \opt / \nnz{A}$ we can compute w.h.p.\ in time 
	$O(\min\{ \nnz{A} + m+n, \psi^{-1} (m+n) \} \cdot \log^3 (mn))$
	vectors $\tu\in\{0,1\}^{m}$ and $\tv\in\{0,1\}^{n}$ 
	such that w.h.p.\
	$\nnzs{A - \tu\cdot \tv^T} \le (1+500\psi) \opt$. 
\end{thm}

\begin{proof}[Proof of Theorem~\ref{thm_paramApproxAlgCombined}]
  We compute a 3-approximation of $\opt$ 
  by applying Theorem~\ref{thm:BinaryRankOneThereeApprox} with $\eps = 0.1$.
  This yields a value $\phi$ satisfying $\psi \le \phi \le 3 \psi$. 
  If $\phi > 1/80$, then the 3-approximation is already good enough, 
  since $\psi > 1/240$ and $1+500\psi > 3$.
  Otherwise, we run Theorem~\ref{thm_paramApproxAlg} with $\phi$.
  Further, using $\phi^2 \nnz{A} \le 9 \psi^2 \nnz{A} = 9 \psi \opt$, the total error is at most 
  \[
  (1+5\phi) \opt + 37 \phi^2 \nnz{A} \le (1+15 \psi) \opt + 37 \cdot 9 \psi \opt \le (1 + 500 \psi) \opt.
  \]
  A rough upper bound on the running time is 
  $O(\min\{ \nnz{A} + m+n, \psi^{-1} (m+n) \} \cdot \log^3 (mn))$.
\end{proof}

A variant of the algorithm from Theorem~\ref{thm_paramApproxAlg} can also be used to solve the \probOne problem exactly. This yields the following theorem, which in particular shows that the problem is in polynomial time whenever $\opt \le O\big(\sqrt{\nnzs{A}} \log(mn)\big)$.

\begin{thm}\label{thm:FPTAlg}
	Given a matrix $A\in\{0,1\}^{m\times n}$, if $\opt_A / \nnz{A} \le 1/240$ then
	we can exactly solve the \probOne problem in time
	$2^{O(\opt / \sqrt{\nnz{A}})} \cdot \mathrm{poly}(mn)$.
\end{thm}

The remainder of this section is devoted to the proof of Theorem~\ref{thm_paramApproxAlg} (and Theorem~\ref{thm:FPTAlg}).

\subsection{Preparations}

Let $A\in\{ 0,1\}^{m\times n}$ and fix an optimal solution $u, v$ to the problem, realizing $\opt = \nnzs{A - u v^{T}}$. Moreover, set $\alpha \Def  \nnz{u}$ and $\beta \Def  \nnz{v}$. 
We start with the following technical preparations.

\begin{lem} \label{lem_technical}
  For any row $i \in [m]$ let $x_i$ be the number of 0's in columns selected by $v$, i.e., $x_i \Def  \{ j \in [n] \mid A_{i,j} = 0, \, v_j = 1\}$, and let $y_i$ be the number of 1's in columns not selected by~$v$, i.e., $y_i \Def  \{ j \in [n] \mid A_{i,j} = 1, \, v_j = 0\}$. 
  Let $R = \{i \in [m] \mid u_i = 1\}$ be the rows selected by $u$, and let $\bar R \Def  [m] \setminus R$. Symmetrically, let $C$ be the columns selected by $v$.
  Then we have
  \begin{enumerate}
    \item $\nnz{A_{i,:}} = \beta - x_i + y_i$ for any $i \in [m]$, \label{line_xyAbeta}
    \item $\opt = \sum_{i \in R} (x_i + y_i) + \sum_{i \in \bar R} (\beta - x_i + y_i)$, \label{line_optxyeq}
    \item $\opt \ge \sum_{i \in R} |x_i - y_i|$, \label{line_optxyineq}
    \item $x_i \le \beta / 2$ for any $i \in R$, and $x_i \ge \beta/2$ for any $i \in \bar R$, \label{line_xbeta}
    \item $\opt \ge \sum_{i=1}^m \min\{ \nnz{A_{i,:}}, |\nnz{A_{i,:}} - \beta | \}$, \label{line_optnnz}
    \item $\left| \nnz{A}-\alpha \beta \right| \le \opt$, \label{line_nnzalphabetaopt}
    \item If $\opt \le \phi \nnz{A}$ then
 $(1-\phi) \nnz{A} \le \alpha \beta \le (1+\phi) \nnz{A}$. \label{line_nnzaphabeta}
  \end{enumerate}
\end{lem}
\begin{proof}
  For (\ref{line_xyAbeta}), note that in the $\beta$ columns $C$ selected by $v$, row $i$ has $\beta - x_i$ 1's, and in the remaining $n-|C|$ columns row $i$ has $y_i$ 1's. Hence, the total number of 1's in row $i$ is $\nnz{A_{i,:}} = \beta - x_i  +y_i$.  

  (\ref{line_optxyeq}) We split $\opt = \nnzs{A - u v^{T}}$ into a sum over all rows, so that $\opt = \sum_{i=1}^m \nnzs{ A_{i,:} - u_i v^T }$. For $i \in \bar R$, the $i$-th term of this sum is simply $\nnz{A_{i,:}} = \beta - x_i + y_i$. For $i \in R$, the $i$-th term is $\nnzs{ A_{i,:} - v^T } = x_i + y_i$.
  
  (\ref{line_optxyineq}) follows immediately from (\ref{line_optxyeq}). 
  
  (\ref{line_xbeta}) follows from (\ref{line_optxyeq}), since for $x_i > \beta/2$ and $i \in R$ we can change $u_i$ from 1 to 0, reducing the contribution of row $i$ from $x_i + y_i$ to $\beta - x_i + y_i$, which contradicts optimality of $\opt$. 
  
  For (\ref{line_optnnz}), we use that 
  $x_i + y_i \ge |x_i - y_i| = |\nnz{A_{i,:}} - \beta|$ by (\ref{line_xyAbeta}),
  and
  \[
  \opt = \sum_{i \in R} (x_i + y_i) + \sum_{i \in \bar R} (\beta - x_i + y_i) = 
  \sum_{i \in R} (x_i + y_i) + \sum_{i \in \bar R} \nnz{A_{i,:}}.
  \] 
  
  (\ref{line_nnzalphabetaopt}) is shown similarly to (\ref{line_optnnz}) by noting that
  \begin{align*}
    \opt &= \sum_{i \in R} (x_i + y_i) + \sum_{i \in \bar R} (\beta - x_i + y_i) \ge \sum_{i \in R} |\nnz{A_{i,:}} - \beta| + \sum_{i \in \bar R} \nnz{A_{i,:}} \\
    &\ge \sum_{i \in R} (\nnz{A_{i,:}} - \beta) + \sum_{i \in \bar R} \nnz{A_{i,:}} = \nnz{A} - \alpha \beta, 
  \end{align*}
  and similarly 
  \[ \opt \ge \sum_{i \in R} |\nnz{A_{i,:}} - \beta| + \sum_{i \in \bar R} \nnz{A_{i,:}} \ge \sum_{i \in R} (\beta - \nnz{A_{i,:}}) - \sum_{i \in \bar R} \nnz{A_{i,:}} = \alpha \beta - \nnz{A}. \]
  
  Finally, (\ref{line_nnzaphabeta}) follows immediately from (\ref{line_nnzalphabetaopt}) by plugging in the upper bound $\opt \le \phi \nnz{A}$.
\end{proof}

\subsection{Approximating $\alpha$ and $\beta$}

We now show how to approximate $\alpha = \nnz{u}$ and $\beta = \nnz{v}$, where $u,v$ is an optimal solution.

\begin{lem} \label{lem_approxalphabeta}
	Given $A\in\{ 0,1\}^{m\times n}$ and $\phi \in (0,1/30]$ with $\opt / \nnz{A} \le \phi$, 
	we can compute in time $O(\nnz{A} + m+n)$ an integer $\tilde \beta \in [m]$ with 
	\[
		\frac{1-3\phi}{1-\phi} \beta \le \tilde \beta \le \frac{1+\phi}{1-\phi} \beta.
	\] 
	Symmetrically, we can approximate $\alpha$ by $\tilde \alpha$. If we are additionally given the number of 1's in each row and column, then the running time becomes $O(m+n)$.
\end{lem}
\begin{proof}
  Let 
  \[ \Lambda \Def  \min_{\beta' \in [m]} \sum_{i=1}^m \min\Big\{ \nnz{A_{i,:}}, \big|\nnz{A_{i,:}} - \beta' \big| \Big\}, \]
  and let $\tilde \beta$ be the value of $\beta'$ realizing $\Lambda$.
  
  We first verify the approximation guarantee. Consider the set of rows $R$ selected by $u$. Let $x_i,y_i$ for $i \in R$ be as in Lemma~\ref{lem_technical}. Then we have
  \[ \Lambda \ge \sum_{i \in R} \min\Big\{ \nnz{A_{i,:}}, \big|\nnz{A_{i,:}} - \tilde \beta \big| \Big\}
  =  \sum_{i \in R} \min\Big\{ \beta + y_i - x_i, |\beta - \tilde \beta + y_i  - x_i | \Big\}, \]
  where we used Lemma~\ref{lem_technical}.\ref{line_xyAbeta}. 
  Assume for the sake of contradiction that 
  $|\beta - \tilde{\beta}| > \tfrac{2\phi}{1-\phi} \beta$.
  Since $|x-y|\geq|x|-|y|$ for any numbers $x,y$, we obtain 
  \[
  |\beta - \tilde \beta + y_i  - x_i | \ge |\beta - \tilde{\beta}| - |x_i - y_i| > \frac{2\phi}{1-\phi} \beta - |x_i - y_i|.
  \]
   Similarly, we have $\beta + y_i - x_i > \tfrac{2\phi}{1-\phi} \beta - |x_i - y_i|$. Hence,
  \[
  	\Lambda > \sum_{i \in R} \left( \frac{2\phi \beta}{1-\phi} - |x_i - y_i| \right) 
  	\ge \frac{2\phi \beta}{1-\phi} |R| - \opt, 
  \]
  where we used Lemma~\ref{lem_technical}.\ref{line_optxyineq}.
  Since $R$ is the set of rows selected by $u$, we have $|R| = \alpha$. By Lemma~\ref{lem_technical}.\ref{line_nnzaphabeta}, we have $\opt \le \phi \nnz{A} \le \tfrac{\phi}{1-\phi} \alpha \beta$. Together, this yields $\Lambda > \opt$, contradicting $$\Lambda \le \sum_{i=1}^m \min\{ \nnz{A_{i,:}}, |\nnz{A_{i,:}} - \beta | \} \le \opt$$ by Lemma~\ref{lem_technical}.\ref{line_optnnz}. Hence, $|\beta - \tilde \beta| \le \tfrac{2\phi}{1-\phi} \beta$.
  
  It remains to design a fast algorithm. We first compute all numbers $\nnz{A_{i,:}}$ in time $O(\nnz{A})$ (this step can be skipped if we are given these numbers as input). We sort these numbers, obtaining a sorted order $\nnzs{A_{\pi(1),:}} \le \ldots \le \nnzs{A_{\pi(m),:}}$. Using counting sort, this takes time $O(m+n)$. We precompute prefix sums 
  $P(k) \Def  \sum_{\ell = 1}^k \nnzs{A_{\pi(\ell),:}}$,
  which allows us to evaluate in constant time any sum 
  \[
  \sum_{\ell=x}^{y} \nnzs{A_{\pi(\ell),:}} = P(y) - P(x-1).
  \]
  
  Finally, we precompute the inverse 
  \[
  \ell(\beta') \Def  \max\{ \ell \mid \nnzs{A_{\pi(\ell),:}} \le \beta' \},
  \]
  or $\ell(\beta')=0$ if there is no $\ell$ with $\nnzs{A_{\pi(\ell),:}} \le \beta'$. By a simple sweep, all values $\ell(\beta')$ can be computed in total time $O(m+n)$.
  
  Note that for any fixed $\beta'$ and row $i$, the term realizing 
  $\min\{ \nnz{A_{i,:}}, |\nnz{A_{i,:}} - \beta' | \}$ is equal to:\\
  (a) $\nnz{A_{i,:}}$ if $\nnz{A_{i,:}} \le \beta'/2$;$\quad$
  (b) $\beta' - \nnz{A_{i,:}}$, if $\beta'/2 < \nnz{A_{i,:}} \le \beta'$;$\quad$ and $\quad$
  (c) $\nnz{A_{i,:}} - \beta'$, if $\nnz{A_{i,:}} > \beta'$.
  Hence, we obtain
  \begin{align*}
   &\sum_{i=1}^n \min\Big\{ \nnz{A_{i,:}}, \big|\nnz{A_{i,:}} - \beta' \big| \Big\}  \\
    &= \left(\sum_{i=1}^{\ell(\beta'/2)}\nnz{A_{\pi(i),:}}\right) + 
    \left(\sum_{i=\ell(\beta'/2)+1}^{\ell(\beta')}\beta'-\nnz{A_{\pi(i),:}}\right)
     + \left(\sum_{i=\ell(\beta')+1}^{n}\nnz{A_{\pi(i),:}}-\beta'\right) \\
	& = P(n) - 2\left[P(\ell(\beta')) - P(\ell(\beta'/2))\right] -
    \left[n+\ell(\beta'/2)-2\ell(\beta')\right]\beta'.
  \end{align*}
  This shows that after the above precomputation the sum $\sum_{j=1}^n \min\{ \nnz{A_{i,:}}, |\nnz{A_{:j}} - \beta' | \}$ can be evaluated in time $O(1)$ for any $\beta'$. Minimizing over all $\beta' \in [m]$ yields $\tilde \beta$. This finishes our algorithm, which runs in total time $O(\nnz{A} + m + n)$, or $O(m+n)$ if we are given the number of 1's in each row and column.
\end{proof}

\subsection{The Algorithm}

We now design an approximation algorithm for the \probOne problem that will yield Theorem~\ref{thm_paramApproxAlg}. 
We present the pseudocode of this Algorithm \ref{alg_pL0rank1} below.

\begin{algorithm}[H]
	\textbf{Input:} $A\in\left\{ 0,1\right\} ^{m\times n}$ and $\phi \in (0,1/80]$ such that
	$\opt / \nnz{A} \le \phi$.\\
	\textbf{Output:} Vectors $\tu \in \{0,1\}^{m}, \, \tv \in \{0,1\}^{n}$ 
	such that $\nnzs{A - \tu \cdot \tv^{T}} \le (1+\eps) \opt$.\\
	
	1. Compute approximations $\frac{1-3\phi}{1-\phi} \alpha \le \tilde \alpha \le \frac{1+\phi}{1-\phi} \alpha$ and $\frac{1-3\phi}{1-\phi} \beta \le \tilde \beta \le \frac{1+\phi}{1-\phi} \beta$ using Lemma~\ref{lem_approxalphabeta}.
	
	Initialize $R^R \Def  [m]$,  $C^R \Def  [n]$, $R^S \Def  \emptyset$, $C^S \Def  \emptyset$.
	\smallskip\smallskip
	
	2. For any row $i$, if $\nnz{A_{i,:}} <\frac{1-\phi}{1+\phi} \cdot \tfrac{\tilde \beta}2$ then set $\tu_{i}=0$ and remove $i$ from $R^R$.
	
	\textcolor{white}{5.} For any column $j$, if $\nnz{A_{:,j}} <\frac{1-\phi}{1+\phi} \cdot \tfrac{\tilde \alpha}2$ then set $\tv_{j}=0$ and remove $j$ from $C^R$.
	\smallskip\smallskip
	
	3. For any $i \in R^R$ compute an estimate $X_i$ with $\big| X_i - \nnzs{A_{i,C^R}} \big| \le \tfrac 19 |C^R|$.
	
	\textcolor{white}{5.} For any $j \in C^R$ compute an estimate $Y_j$ with $\big| Y_j - \nnzs{A_{R^R,j}} \big| \le \tfrac 19 |R^R|$.
	\smallskip\smallskip	
	
	4. For any $i \in R^R$, if $X_i > \tfrac 23 \tilde \beta$ then set $\tu_{i}=1$ and add $i$ to $R^S$.
	
	\textcolor{white}{5.} For any $j \in C^R$, if $Y_j > \tfrac 23 \tilde \alpha$ then set $\tv_{j}=1$ and add $j$ to $C^S$.
	\smallskip\smallskip	
	
	5. For any $i \in R^R \setminus R^S$, compute an estimate $X'_i$ with $| X'_i - \nnzs{A_{i,C^S}} | \le \phi |C^S|$,
	
	\textcolor{white}{5.} For any $j \in C^R \setminus C^S$, compute an estimate $Y'_j$ with $| Y'_j - \nnzs{A_{R^S,j}} | \le \phi |R^S|$.
	\smallskip\smallskip	
	
	6. For any $i \in R^R \setminus R^S$, set $\tu_i = 1$ if $X'_i \ge |C^S|/2$ and 0 otherwise,
	
	\textcolor{white}{5.} For any $j \in C^R \setminus C^S$, set $\tv_j = 1$ if $Y'_j \ge |R^S|/2$ and 0 otherwise.
	\smallskip\smallskip	
	
	7. \textbf{Return} $(\tu,\tv)$.
	
	\caption{Boolean $\ell_0$-Rank-$1$ With Small Optimal Value}
	\label{alg_pL0rank1}
\end{algorithm}

\paragraph{Running Time}
By Lemma~\ref{lem_approxalphabeta}, Step 1 runs in time $O(\nnz{A} + m+n)$, or in time $O(m+n)$ if we are given the number of 1's in each row and column. Steps 2, 4, and 6 clearly run in time $O(m+n)$. For steps 3 and 5, there are two ways to implement them. 

Variant (1) is an exact algorithm. We enumerate all nonzero entries of $A$ and count how many contribute to the required numbers $\nnzs{A_{i,C^R}}, \nnzs{A_{R^R,j}}$ etc. This takes total time $O(\nnz{A})$, and hence the total running time of the algorithm is $O(\nnz{A} + m+n)$. 

Variant (2) uses random sampling. In order to estimate $\nnzs{A_{i,C^R}}$, consider a random variable $Z$ that draws a uniformly random column $j \in C^R$ and returns 1 if $A_{i,j} \ne 0$ and 0 otherwise. Then $\mathbb{E}[Z] = \nnzs{A_{i,C^R}} / |C^R|$. Taking independent copies $Z_1,\ldots,Z_\ell$ of $Z$, where $\ell = \Theta( \log(mn) / \delta^2)$ with sufficiently large hidden constant, a standard Chernoff bound argument shows that w.h.p.\ 
\[
	\left| (Z_1+\ldots+Z_\ell) \cdot \frac{|C^R|}{\ell} - \nnzs{A_{i,C^R}}\right| \le \delta\cdot |C^R|,
\]
which yields the required approximation. For Step 3 we use this procedure with $\delta = \tfrac 19$ and obtain running time $O(\log (mn))$ per row and column, or $O((m+n) \log (mn))$ in total. For Step 5 we use $\delta = \phi$, resulting in time $O(\phi^{-2}\log(mn))$ for computing one estimate $X'_i$ or $Y'_j$. By Claim~\ref{clm_Step3_Cnt} below there are only $O(\phi(m+n))$ rows and columns in $R^R \setminus R^S$ and $C^R \setminus C^S$, and hence the total running time for Step 5 is $O(\phi^{-1}(m+n)\log(mn))$. This dominates the total running time.

Combining both variants, we obtain the claimed running time of
\[
	O(\min\{\nnz{A} + m+n, \phi^{-1} (m+n) \log (mn)\}).
\]
 
\newpage
\paragraph{Correctness}
In the following we analyze the correctness of the above algorithm.

\begin{clm}\label{clm_Step2_Correct} 
	For any row $i$ deleted in Step 2 we have $\tu_i = u_i$. Symmetrically, for any column $j$ deleted in Step 2 we have $\tv_j = v_j$.
\end{clm}
\begin{proof}
	If row $i$ is deleted, then by the approximation guarantee of $\tilde \beta$ we have
	\[
	\nnz{A_{i,:}} <\frac{1-\phi}{1+\phi} \cdot \frac{\tilde\beta}{2} \le \frac{\beta}{2}
	\] 
	Note that for $x_i$ (the number of 0's in row $i$ in columns selected by $v$) we have $x_i \ge \beta - \nnz{A_{i,:}}$. Together, we obtain $x_i > \beta/2$, and thus row $i$ cannot be selected by $u$, by Lemma~\ref{lem_technical}.\ref{line_xbeta}. Hence, we have $u_i = 0 = \tu_i$.
	The statement for the columns is symmetric.
\end{proof}

\begin{clm}\label{clm_Step2_Cnt} 
	After Step 2, it holds for the remaining rows $R^{R}$ and columns $C^{R}$ that
	\[
	|R^{R}|\leq \left(1 + \frac{1+\phi}{1-3\phi} \cdot \frac{2\phi}{1-\phi} \right) \alpha
	\quad\text{ and }\quad
	|C^{R}|\leq \left(1 + \frac{1+\phi}{1-3\phi} \cdot \frac{2\phi}{1-\phi} \right) \beta.
	\]
\end{clm}
\begin{proof}
	By Claim \ref{clm_Step2_Correct} the $\alpha$ rows $R$ selected by $u$
	remain. We split the rows $R^R$ remaining after Step 2 into $R \cup R'$, 
	and bound $|R'|$ from above. Since any $i \in R'$ is not selected by $u$, 
	it contributes $\nnz{A_{i,:}}$ to $\opt$. Note that 
	\[
	\nnz{A_{i,:}} \ge \frac{1-\phi}{1+\phi} \cdot \frac {\tilde \beta}2 \ge \frac{1-\phi}{1+\phi} \cdot \frac{1-3\phi}{1-\phi} \cdot \frac \beta 2 = \frac{1-3\phi}{1+\phi} \cdot \frac \beta 2,
	\]
	and thus $|R'| \le \opt \cdot \frac{1+\phi}{1-3\phi} \cdot \frac 2 \beta$.
	Since 
	\[
	\opt \le \phi \nnz{A} \le \frac{\phi}{1-\phi}\cdot \alpha \beta
	\]
	by Lemma~\ref{lem_technical}.\ref{line_nnzaphabeta}, we obtain $|R'| \le \frac{1+\phi}{1-3\phi} \cdot \tfrac{2\phi}{1-\phi} \cdot \alpha$. Thus, we have in total
	\[
	|R^R| = |R| + |R'| \le 
	\left(1 + \frac{1+\phi}{1-3\phi} \cdot \frac{2\phi}{1-\phi} \right) \alpha.	
	\]
	The statement for the columns is symmetric.
\end{proof}

\begin{clm}
	\label{clm_Step3_Correct} The rows and columns selected in Step 4
	are also selected by the optimal solution $u,v$, i.e., for any $i \in R^S$ 
	we have $u_i = 1$ and for any $j \in C^S$ we have $v_j = 1$.
\end{clm}
\begin{proof}
	If row $i$ is selected in Step 4, then we have by the approximation guarantee of $X_i$, definition of Step 4, Claim~\ref{clm_Step2_Cnt}, and Lemma~\ref{lem_approxalphabeta}
	\begin{align*}
	\nnzs{A_{i,C^R}} 
	&\ge X_i - \frac{1}{9} |C^R| > \frac 23 \tilde \beta - 
	\frac{1}{9} \left(1 + \frac{1+\phi}{1-3\phi} \cdot \frac{2\phi}{1-\phi} \right) \beta \\
	& \ge \frac{2}{3} \cdot \frac{1-3\phi}{1-\phi} \beta - 
	\frac{1}{9} \left(1 + \frac{1+\phi}{1-3\phi} \cdot \frac{2\phi}{1-\phi} \right) \beta.
	\end{align*}
	It is easy to see that for sufficiently small $\phi \ge 0$ this yields
	\[
	\nnzs{A_{i,C^R}} > \frac \beta 2 + \frac{1+\phi}{1-3\phi} \cdot \frac{2\phi}{1-\phi} \beta.
	\]
	One can check that $0 \le \phi \le 1/80$ is sufficient. 
	Since there are $|C^{R}|\leq \big(1 + \frac{1+\phi}{1-3\phi} \cdot \tfrac{2\phi}{1-\phi} \big) \beta$ columns remaining, in particular the $\beta$ columns $C \subseteq C^R$ which are selected by $v$, we obtain 
	\[ 
	\nnzs{A_{i,C}} \ge \nnzs{A_{i,C^R}} - (|C^R| - \beta) > \beta /2. 
	\]
	By Lemma~\ref{lem_technical}.\ref{line_xbeta}, we thus obtain that row $i$ is selected by the optimal $u$.
	The statement for the columns is symmetric.
\end{proof}

\begin{clm}
	\label{clm_Step3_Cnt} 
	After Step 4 there are $|R^R \setminus R^S| \le 6 \phi \alpha$ remaining rows 
	and $|C^R \setminus C^S| \le 6 \phi \beta$ remaining columns.
\end{clm}
\begin{proof}
	After Step 4, every remaining row $i$, for any $0 \le \phi \le 1/80$, satisfies 
	\[
	\nnz{A_{i,:}} \ge \frac{1-\phi}{1+\phi} \cdot \frac{\tilde \beta}2 \ge \frac{1-\phi}{1+\phi} \cdot \frac{1-3\phi}{1-\phi} \cdot \frac{\beta}2 \ge \frac 25 \beta,
	\]
	Moreover, each such row satisfies
	\[
	\nnzs{A_{i,C^R}}\le X_i + \frac 19 |C^R| \le \frac 23 \tilde \beta + \frac 19 |C^R|,
	\]
	which together with $\tilde \beta \le \tfrac{1+\phi}{1-\phi} \beta$ (Lemma~\ref{lem_approxalphabeta}) and $|C^R| \le \big(1 + \frac{1+\phi}{1-3\phi} \cdot \tfrac{2\phi}{1-\phi} \big) \beta$ (Claim~\ref{clm_Step2_Cnt}) yields 
	\[
	\nnzs{A_{i,C^R}}\le \left(\frac 23 \cdot \frac{1+\phi}{1-\phi} + \frac 19 + \frac 19 \cdot \frac{1+\phi}{1-3\phi} \cdot \frac{2\phi}{1-\phi} \right) \beta.
	\]
	It is easy to see that for sufficiently small $\phi \ge 0$ we have $\nnzs{A_{i,C^R}}\le \tfrac 45 \beta$, and it can be checked that $0 \le \phi \le 1/80$ is sufficient.
	
	If $i$ is not selected by $u$, then its contribution to $\opt$ is $\nnz{A_{i,:}} \ge \tfrac 25 \beta$. 
	If $i$ is selected by $u$, then since $C \subseteq C^R$ its contribution to $\opt$ is at least
	\[
	\beta - \nnz{A_{i,C}} \ge \beta - \nnzs{A_{i,C^R}} \ge \beta - \frac 45 \beta = \frac 15 \beta.
	\]
	Thus, the number of remaining rows is at most 
	\[
	\frac \opt {\beta/5} \le \frac{ 5 \phi \alpha \beta}{(1-\phi) \beta} \le 6 \phi \alpha,
	\] 
	where we used Lemma~\ref{lem_technical}.\ref{line_nnzaphabeta}. 
	The statement for the columns is symmetric.
\end{proof}

We are now ready to prove correctness of Algorithm~\ref{alg_pL0rank1}.
\begin{proof}
	[Proof of Theorem~\ref{thm_paramApproxAlg}] 
	The rows and columns removed in Step 2 are also not picked by the optimal solution, by Claim \ref{clm_Step2_Correct}. Hence, in the region $([m] \setminus R^R) \times [n]$ and $[m] \times ([n] \setminus C^R)$ we incur the same error as the optimal solution.
	The rows and columns chosen in Step 4 are also picked by the optimal solution, by Claim \ref{clm_Step3_Correct}. Hence, in the region $R^S \times C^S$ we incur the same error as the optimal solution. We split the remaining matrix into three regions: $(R^R \setminus R^S)\times C^S$, $R^S \times(C^R \setminus C^S)$, and $(R^R \setminus R^S) \times(C^R \setminus C^S)$. 
	
	In the region $(R^R \setminus R^S)\times C^S$ we compute for any row $i \in R^R \setminus R^S$ an additive $\phi |C|$-approximation $X'_i$ of $\nnzs{A_{i,C}}$, and we pick row $i$ iff $X'_i \ge |C|/2$. In case $\big| \nnzs{A_{i,C}} - |C|/2 \big| > \phi |C|$, we have $X'_i \ge |C|/2$ if and only if $\nnzs{A_{i,C}} \ge |C|/2$, and thus our choice for row $i$ is optimal, restricted to region $(R^R \setminus R^S)\times C^S$. Otherwise, if $\big| \nnzs{A_{i,C}} - |C|/2 \big| \le \phi |C|$, then no matter whether we choose row $i$ or not, we obtain approximation ratio 
	\[
	\frac{|C|/2 + \phi |C|}{|C|/2 - \phi |C|} = \frac{1+2\phi}{1-2\phi} \le 1+5\phi,
	\]
	restricted to region $(R^R \setminus R^S)\times C^S$. 
	The region $R^S \times(C^R \setminus C^S)$ is symmetric.
	
	Finally, in region $(R^R \setminus R^S) \times(C^R \setminus C^S)$ we pessimistically assume that every entry is an error. 
	By Claim~\ref{clm_Step3_Cnt} and Lemma~\ref{lem_technical}.\ref{line_nnzaphabeta}, 
	this submatrix has size at most 
	\[
	6 \phi \alpha \cdot 6 \phi \beta \le 36 \phi^2 (1+\phi) \nnz{A} \le 37 \phi^2 \nnz{A}.
	\]
	In total, over all regions, we computed vectors $\tu, \tv$ such that
	\[
	\nnzs{A - \tu \tv^T} \le (1+5\phi) \opt + 37 \phi^2 \nnz{A}.
	\]
	This completes the correctness prove of Algorithm~\ref{alg_pL0rank1}. 
\end{proof}

\subsection{The Exact Algorithm}

We now prove Theorem~\ref{thm:FPTAlg}, i.e., given a matrix $A\in\{0,1\}^{m\times n}$ we exactly solve the \probOne problem in time $2^{O(\opt / \sqrt{\nnz{A}})} \cdot\mathrm{poly}(mn)$ if we have $\psi \Def  \opt / \nnz{A} \le 1/240$. 
This algorithm builds upon the algorithmic
results established in Theorem~\ref{thm_paramApproxAlgCombined} and 
Theorem~\ref{thm_ImprRank1}, and it consists of the following three phases:

1.  Run the algorithm in Theorem~\ref{thm_ImprRank1} to compute 
a $3$-approximation of $\psi=\mathrm{OPT}/\lVert A\rVert_{0}$, 
i.e. a number $\phi \in [\psi, 3\psi]$.

2. Run Steps 1-4 of Algorithm~\ref{alg_pL0rank1}, resulting in selected rows $R^S$ and columns $C^S$, and undecided rows $R' = R^R \setminus R^S$ and columns $C' = C^R \setminus C^S$. As shown above, the choices made by these steps are optimal.

3. For the remaining rows $R'$ and
columns $C'$ we use brute force to
find the optimum solution. Specifically, assume without loss of generality that $|R'| \le |C'|$. Enumerate all Boolean vectors $u' \in \{0,1\}^{R'}$. For each $u'$, set $\tilde u_i = u'_i$ for all $i \in R'$ to complete the specification of a vector $\tilde u \in \{0,1\}^m$. We can now find the optimal choice of vector $\tilde v$ in polynomial time, since the optimal choice is to set $\tilde v_j = 1$ iff column $A_{:,j}$ has more 1's than 0's in the support of $\tilde u$. Since some $u'$ gives rise to the optimal vector $\tilde u = u$, we solve the \probOne problem exactly.

To analyze the running time, note that by Claim~\ref{clm_Step3_Cnt} we have 
\[
\min\{|R'|,|C'|\} \le 6 \phi \min\{\alpha,\beta\} \le 6 \phi \sqrt{\alpha \beta}.
\]
By Lemma~\ref{lem_technical}.\ref{line_nnzaphabeta} and $\phi \le 3 \psi$, we obtain $\min\{|R'|,|C'|\} = O(\psi \sqrt{\nnz{A}})$. Hence, we enumerate $2^{O(\psi \sqrt{\nnz{A}})} = 2^{O(\opt / \sqrt{\nnz{A}})}$ vectors $u'$, and the total running time is $2^{O(\opt / \sqrt{\nnz{A}})} \cdot \mathrm{poly}(mn)$. 
This completes the proof of Theorem~\ref{thm:FPTAlg}.

\newpage
\section{Sample Complexity Lower Bound for Boolean $\ell_{0}$-Rank-$1$}

We now give a lower bound of $\Omega(n/\phi)$ on the number of samples 
of any $(1+O(\phi))$-approximation algorithm for the \probOne problem, 
where $\phi \ge \opt / \nnz{A}$ as before.

\begin{theorem}\label{thm:SampleLowerBoundBRone}
	Let $C \ge 1$.
	Given an $n \times n$ binary matrix $A$ with column adjacency arrays and 
	with row and column sums, and given $\sqrt{\log (n) / n} \ll \phi \leq 1/100C$ 
	such that $\opt / \nnz{A} \le \phi$, computing a $(1+C \phi)$-approximation of
	the \probOne problem requires to read $\Omega(n/\phi)$ entries of $A$ (in the worst case over $A$). 
\end{theorem}

\subsection{Core Probabilistic Result}

The technical core of our argument is the following lemma.

\begin{lem} \label{lem:estimatingexpectations}
	Let $\phi\in(0,1/2)$. Let $X_1, \ldots, X_k$ be binary random variables with expectations $p_1, \ldots, p_k$, where $p_i \in \{1/2 -\phi, 1/2+\phi\}$ for each $i$. Let $\A$ be an algorithm which can adaptively obtain any number of samples of each random variable, and which outputs bits $b_i$ for every $i \in [1:k]$. Suppose that with probability at least $0.95$ over the joint probability space of $\A$ and the random samples, $\A$ outputs for at least a $0.95$ fraction of all $i$ that 
	$b_i = 1$ if $p_i = 1/2 + \phi$ and $b_i = 0$ otherwise. 
	Then, with probability at least $0.05$, $\A$ makes $\Omega(k/\phi^2)$ samples in total, asymptotically in $k$.
\end{lem}

\begin{proof}
	
	Consider the following problem $P$: let $X$ be a binary random variable with expectation
	$p$ drawn uniformly in $\{1/2-\phi, 1/2+\phi\}$. It is well-known that any algorithm which, with probability at least $0.6$, obtains samples from $X$ and outputs $0$ if $p = 1/2-\phi$ and outputs $1$ if $p = 1/2 +\phi$, requires $\Omega(1/\phi^2)$ samples; see, e.g., Theorem 4.32 of \cite{b02}. Let $c > 0$ be such that $c/\phi^2$ is a lower bound on the number of samples for this problem $P$. 
	
	Let $\A$ be an algorithm solving the problem in the lemma statement. Since $\A$ succeeds with probability at least $0.95$ in obtaining the guarantees of the lemma for given sequence $p_1, \ldots, p_k$, it also succeeds with this probability when $(p_1, \ldots, p_k)$ is drawn from the uniform distribution on $\{1/2-\phi, 1/2+\phi\}^k$. 
	
	Suppose, towards a contradiction, that $\A$ takes less than $0.05 \cdot c k/\phi^2$ samples with probability at least $0.95$. By stopping $\A$ before taking $0.05 \cdot c k/\phi^2$ samples, we obtain an algorithm $A'$ that always takes less than $0.05 \cdot c k/\phi^2$ samples. 
	By the union bound, $A'$ obtains the guarantees of the lemma for the output bits $b_i$ with probability at least $0.9$, over the joint probability space of $A'$ and the random samples.
	
	Note that the expected number of samples $A'$ takes from a given $X_i$ is less than $0.05 \cdot c/\phi^2$. By Markov's inequality, for a $0.95$ fraction of indices $i$, $A'$ takes less than $c/\phi^2$ samples from $X_i$. We say that $i$ is {\it good} if $A'$ takes less than $c/\phi^2$ samples from $X_i$ and the output bit $b_i$ is correct. By union bound, at least a $1-(1-0.9)-(1-0.95)=0.85$ fraction of indices $i$ is good. 
	
	Since $(p_1, \ldots, p_k)$ is drawn from the uniform distribution on $\{1/2-\phi, 1/2+\phi\}^k$, 
	with probability at least $0.95$ the number $k_+ = |\{ i \,:\, p_i = 1/2 + \phi \}|$ satisfies 
	$0.45 k \le k_+ \le 0.55 k$ (for sufficiently large $k$). This implies that a $0.65$ fraction 
	of indices 
	$\{  i \,:\, p_i = 1/2 + \phi \}$ is good, as otherwise the number of bad $i$'s is at least 
	$(1-0.65) \cdot 0.45 k > 0.15 k$. 
	Similarly, a $0.65$ fraction of indices $\{  i \,:\, p_i = 1/2 - \phi \}$ is good. 
	
	Given an instance of problem $P$ with random variable $X$ and expectation $p$, we choose a uniformly random $i \in [k]$, and set $X_i = X$. For $j \neq i$, we independently and uniformly at random choose $p_j \in \{1/2-\phi, 1/2+\phi\}$. We then run algorithm $A'$. Whenever $A'$ samples from $X_i$, we sample a new value of $X$ as in problem $P$. Whenever $A'$ samples from $X_j$ for $j \neq i$, we flip a coin with probability $p_j$ and report the output to $A'$. If $A'$ takes $c/\phi^2$ samples from $X_i$, then we abort, thus ensuring that $A'$ always takes less than $c/\phi^2$ samples from $X_i=X$. Observe that the input to $A'$ is a sequence of random variables $X_1, \ldots, X_k$ with expectations $p_1, \ldots, p_k$ which are independent and uniformly distributed in $\{1/2-\phi, 1/2+\phi\}$. In particular, except for their expectation these random variables are indistinguishable.
	
	We now condition on $0.45 k \le k_+ \le 0.55 k$, which has success probability at least $0.95$ for sufficiently large $k$. Then no matter whether $p_i = 1/2 + \phi$ or $p_i = 1/2 - \phi$, at least a $0.65$ fraction of indices $j$ with $p_j = p_i$ is good. Since $i$ was chosen to be a uniformly random position independently of the randomness of the sampling and the algorithm $A'$, and the $X_j$ with $p_j = p_i$ are indistinguishable, with probability at least $0.65$ index $i$ is good. In this case, $A'$ takes less than $c/\phi^2$ samples from $X_i = X$ and correctly determines the output bit $b_i$, i.e., whether $p_i = 1/2+\phi$. 
	As by union bound the total success probability is $1 - (1-0.65) - (1-0.95) = 0.6$, this contradicts the requirement of $c/\phi^2$ samples mentioned above for solving $P$. Hence, the assumption was wrong, and $\A$ takes $\Omega(k/\phi^2)$ samples with probability at least $0.05$.
\end{proof}

We start with a simplified version of our result, where we only have random access to the matrix entries. Below we extend this lower bound to the situation where we even have random access to the adjacency lists of all rows and columns.

\begin{thm}\label{thm:simpleLB}
	Let $C \ge 1$.
	Given an $n \times n$ binary matrix $A$ by random access to its entries, 
	and given $\sqrt{\log (n) / n} \ll \phi \leq 1/100C$ such that $\opt / \nnz{A} \le \phi$, 
	computing a $(1+C \phi)$-approximation of $\opt$ requires to read $\Omega(n/\phi)$ entries
	of $A$ (in the worst case over $A$). 
\end{thm}
\begin{proof}
	Set $\phi' \Def  25 C \phi$ and $k \Def  \phi n / 2$. As in Lemma~\ref{lem:estimatingexpectations}, consider binary random variables $X_1,\ldots,X_{k}$ with expectations $p_1,\ldots,p_{k}$, where $p_i \in \{1/2 -\phi', 1/2+\phi'\}$ for each $i$. We (implicitly) construct an $n \times n$ matrix $A$ as follows. 
	For ever $k < i \le n, \, 1 \le j \le n$ we set $A_{i,j} \Def  1$. For any $1 \le i \le k, \, 1 \le j \le n$ we sample a bit $b_{i,j}$ from $X_i$ and set $A_{i,j} \Def  b_{i,j}$. Note that we can run any \probOne algorithm implicitly on $A$: whenever the algorithm reads an entry $A_{i,j}$ we sample a bit from $X_i$ to determine the entry (and we remember the entry for possible further accesses). 
	
	Let us determine the optimal solution for $A$. Note that for each $i > k$, since the row $A_{i,:}$ is all-ones, it is always better to pick this row than not to pick it, and thus without loss of generality any solution $u,v$ has $u_i = 1$. Similarly, for any $j$, since the column $A_{:,j}$ has $n-k > n/2$ 1's in rows picked by $u$, it is always better to pick the column than not to pick it, and thus $v_j = 1$, i.e., $v$ is the all-ones vector. Hence, the only choice is for any $1 \le i \le k$ to pick or not to pick row $i$. 
	Note that no matter whether we pick these rows or not, the total error is at most $\phi n^2 / 2$, since these rows in total have $k n = \phi n^2 / 2$ entries, and all remaining entries of $A$ are correctly recovered by the product $u v^T$ by the already chosen entries of $u$ and $v$. Hence, $\opt \le \phi n^2 / 2$, and since $\nnz{A} \ge (n-k)n \ge n^2/2$, we obtain, as required, $\opt / \nnz{A} \le \phi$.
	
	Now consider the rows $1 \le i \le k$ more closely. Since $v$ is the all-ones vector, not picking row $i$ incurs cost for each 1 in the row, which is cost $\nnzs{A_{i,:}}$, while picking row $i$ incurs cost for each 0 in the row, which is cost $n - \nnzs{A_{i,:}}$. Note that the expected number of 1's in row $1 \le i \le k$ is $p_i n$.  The Chernoff bound yields concentration: We have w.h.p.\ 
	$| \nnzs{A_{i,:}} - p_i n | \le 0.01 \cdot \phi' n$, where we used $\phi' \gg \sqrt{\log (n) / n}$. In the following we condition on this event and thus drop ``w.h.p.'' from our statements.
	In particular, for any $i$ with $p_i = 1/2 + \phi'$ we have $\nnzs{A_{i,:}} \ge (1/2 + 0.99 \phi') n$, and for any $i$ with $p_i = 1/2 - \phi'$ we have $\nnzs{A_{i,:}} \le (1/2 - 0.99 \phi') n$. 
	
	By picking all rows $i \le k$ with $p_i = 1/2 + \phi'$ and not pick the rows with $p_i = 1/2 - \phi'$, we see that $\opt \le (1/2 - 0.99 \phi') k n$.
	Now consider a solution $u$ that among the rows $1 \le i \le k$ with $p_i = 1/2 + \phi'$ picks $g_+$ many and does not pick $b_+$ many. Similarly, among the rows with $p_i = 1/2 - \phi'$ it picks $g_-$ and does not pick $b_-$. 
	Note that each of the $g_+$ ``good'' rows incurs cost 
	\[
	n - \nnz{A_{i:}} \ge n - (1/2 + 1.01\phi') n = (1/2 - 1.01\phi') n.
	\]
	Each of the $b_+$ ``bad'' rows incurs a cost of $\nnzs{A_{i:}} \ge (1/2 + 0.99\phi') n$. Similar statements hold for $g_-$ and $b_-$, and thus for $g \Def  g_+ + g_-$ and $b \Def  b_+ + b_-$, with $g+b = k$, we obtain a total cost of 
	\begin{align*}
	\nnzs{A - u v^T} 
	&\ge  g \cdot (1/2 - 1.01 \phi')n + b \cdot (1/2 + 0.99 \phi')n \\
	& = k(1/2-0.99\phi')n + 2b\phi' n - 0.02k\phi'n \\
	&\ge \opt + 2 b \phi' n - 0.02 \phi' k n.
	\end{align*}
	If $b \ge 0.02 k$, then 
	\[
	\nnzs{A - u v^T} \ge \opt + 0.02 \phi' k n \ge (1 + 0.04 \phi') \opt.
	\]
	By contraposition, if we compute a $(1 + 0.04 \phi' = 1+C \phi)$-approximation on $A$, then $b \le 0.02 k$, and thus the vector $u$ correctly identifies for at least a $0.98$ fraction of the random variables $X_i$ whether $p_i = 1/2 + \phi'$ or $p_i = 1/2 - \phi'$. Since this holds w.h.p., by Lemma~\ref{lem:estimatingexpectations} we need $\Omega(k/\phi'^2) = \Omega(n/(\phi C^2))$ samples from the variables $X_i$, and thus $\Omega(n/(\phi C^2))$ reads in $A$. Since $C \ge 1$ is constant, we obtain a lower bound of $\Omega(n/\phi)$. This lower bound holds in expectation over the constructed distribution of $A$-matrices, and thus also in the worst case over $A$.
\end{proof}

\subsection{Hard Instance}

The construction in Theorem~\ref{thm:simpleLB} does not work in the case when 
we have random access to the adjacency lists of the rows, 
since this allows us to quickly determine the numbers of 1's per row, 
which is all we need to know in order to decide whether we pick a particular row 
in the matrix constructed above. 
We overcome this issue, using the adapted construction in Theorem~\ref{thm:SampleLowerBoundBRone}.

We now present the correctness of Theorem~\ref{thm:SampleLowerBoundBRone}.

\begin{proof}[Proof of Theorem~\ref{thm:SampleLowerBoundBRone}]
	We assume that $n$ is even. 
	Let $\phi', k, X_1,\ldots,X_k, p_1, \ldots ,p_k$ be as in the proof of Theorem~\ref{thm:simpleLB}.
	We adapt the construction of the matrix $A$ as follows. 
	For any $2k < i \le n, \, 1 \le j \le n/2$ we set $A_{i,2j} \Def  1$ and $A_{i,2j-1} \Def  0$. For any $1 \le i \le k, \, 1 \le j \le n/2$ we sample a bit $b_{i,j}$ from $X_i$ and set $A_{2i,2j} \Def  A_{2i-1,2j-1} \Def  b_{i,j}$ and $A_{2i-1,2j} \Def  A_{2i,2j-1} \Def  1 - b_{i,j}$. 
	
	As before, when running any \probOne algorithm on $A$ we can easily support random accesses to entries $A_{i,j}$, by sampling from $X_{\lceil i/2 \rceil}$ to determine the entry (and remembering the sampled bit for possible further accesses). Furthermore, we can now allow random accesses to the \emph{adjacency arrays} of rows and columns. Specifically, if we want to determine the $\ell$-th 1 in row~$i \le 2k$, we know that among the entries $A_{i,1},\ldots,A_{i,2\ell}$ there are exactly $\ell$ 1's, since by construction $A_{i,2j-1} + A_{i,2j} = 1$. Hence, the $\ell$-th 1 in row $i$ is at position $A_{i,2\ell-1}$ or $A_{i,2\ell}$, depending only on the sample $b_{\lceil i/2 \rceil,\ell}$ from $X_{\lceil i/2 \rceil}$. For rows $i > 2k$, the $\ell$-th 1 is simply at position $A_{i,2\ell}$.
	Thus, accessing the $\ell$-th 1 in any row takes at most one sample, so we can simulate any algorithm on $A$ with random access to the adjacency lists of rows. The situation for columns is essentially symmetric. Similarly, we can allow constant time access to the row and column sums.
	
	In the remainder we show that the constructed matrix $A$ has essentially the same properties as the construction in Theorem~\ref{thm:simpleLB}. We first argue that any 2-approximation $u,v$ for the \probOne problem on $A$ picks all rows $i > 2k$ and picks all even columns and does not pick any odd column. Thus, the only remaining choice is which rows $i \le 2k$ to pick. 
	To prove this claim, first note that any solution following this pattern has error at most $2k n = \phi n^2$, 
	since the $2k$ undecided rows have $2k n$ entries, and all other entries are correctly recovered by the 
	already chosen parts of $u v^T$.
	Hence, we have $\opt \le \phi n^2$.
	Now consider any 2-approximation $u,v$, which must have cost at most $2\phi n^2$.
	Note that $u$ picks at least $(1-5\phi) n$ of the rows $\{2k+1,\ldots,n\}$, since each such row contains $n/2$ 1's that can only be recovered if we pick the row, so we can afford to ignore at most $8k=4 \phi n$ of these $n-2k = (1-\phi) n$ rows. 
	Now, each even column contains at least $(1-5\phi) n > n/2$ 1's in picked rows, and thus it is always better to pick the even columns. Similarly, each odd column contains at least $n/2$ 0's in picked rows, and thus it is always better not to pick the odd columns. Hence, we obtain without loss of generality $v_{2j} = 1$ and $v_{2j-1} = 0$. Finally, each row $i > 2k$ contains $n/2$ 1's in columns picked by $v$ and $n/2$ 0's in columns not picked by $v$, and thus it is always better to pick row $i$. Hence, we obtain without loss of generality $u_i = 1$ for $i > 2k$. 
	
	Our goal now is to lower bound $\nnzs{A - u v^T}$ in terms of $\opt$ 
	and the error term $b\phi'n$, similarly to the proof in Theorem~\ref{thm:simpleLB}.
	Notice that we may ignore the odd columns, as they are not picked by $v$.
	Restricted to the even columns, row $2i$ is exactly as row $i$ in the construction 
	in Theorem~\ref{thm:simpleLB}, while row $2i-1$ is row $2i$ negated. 
	Thus, analogously as in the proof of Theorem~\ref{thm:simpleLB}, 
	we obtain w.h.p.\ $\opt \le (1/2 - 0.99\phi') 2 k n$ and
	\[
	\nnzs{A - u v^T} \ge \opt + 2b\phi'n - 0.04 k \phi' n\geq(1+0.04\phi')\opt,
	\] 
	where $b\geq0.04k$ is the number of ``bad'' rows $i \le 2k$. 
	Again analogously, if we compute a $(1+0.04 \phi') = 1+C\phi)$-approximation on $A$,
	then $b \le 0.04k$, and thus w.h.p.\ we correctly identify  for at least a $0.96$ 
	fraction of the random variables $X_i$ whether 
	$p_i = 1/2 + \phi'$ or $p_i  =1/2 - \phi'$. 
	As before, this yields a lower bound of $\Omega(n/\phi)$ samples.
\end{proof}

\bibliographystyle{alpha}
\bibliography{ref}

\end{document}